\documentclass[11pt,a4paper,final]{article}

\usepackage{amsmath}    % AMS math style
\usepackage{amsfonts}   % Include AMS fonts
\usepackage{amssymb}    % Include AMS symbols
\usepackage{amsthm}     % For theorem-like environments
\usepackage{mathrsfs}   % Formal Script Math Symbol font
\usepackage{parskip}    % For whitespaced paragraphs
\usepackage{calc}       % For use of \widthof in def of \convas macro
\usepackage{enumitem}   % For nicer spacing in itemized environments
\usepackage{graphicx}   % For \includegraphics
\usepackage{lpic}       % For latex-annotated graphics
\usepackage{color}      % For coloring text
\usepackage{mathtools}
\usepackage{thmtools,thm-restate} % To restate theorems in appendix
\usepackage{subcaption}
\usepackage{placeins}
\usepackage{accents}

\renewcommand{\baselinestretch}{1.25}
\let\Oldsection\section
\renewcommand{\section}{\FloatBarrier\Oldsection}
\let\Oldsubsection\subsection
\renewcommand{\subsection}{\FloatBarrier\Oldsubsection}
\let\Oldsubsubsection\subsubsection
\renewcommand{\subsubsection}{\FloatBarrier\Oldsubsubsection}

\newcommand{\ie}{{\it i.e.}}
\newcommand{\cf}{{\it c.f.}}
\newcommand{\eg}{{\it e.g.}}

\newcommand{\iid}{{\it i.i.d.}}

\newcommand{\PP}{{\mathbb P}}

\newcommand{\RR}{{\mathbb R}}

\newcommand{\NN}{{\mathbb N}}
\newcommand{\ZZ}{{\mathbb Z}}

\newcommand{\FB}{\mathbb{F}_B}
\newcommand{\FA}{\mathbb{F}_A}
\newcommand{\DB}{\mathbb{D}_{B}}
\newcommand{\DA}{\mathbb{D}_{A}}
\newcommand{\Bin}{\mathrm{Bin}}

\newcommand{\calX}{\mathcal{X}}

\newcommand{\ft}[2]{{\textstyle{\frac{#1}{#2}}}}

\newcommand{\conv}[1]%
  {{\mathrel{\,\xrightarrow{\widthof{\,#1\,}}\,}}}
\newcommand{\convas}[1]%
  {{\mathrel{\,\xrightarrow{\widthof{\,#1\text{-a.s.}\,}}\,}}}
\newcommand{\convprob}[1]%
  {{\mathrel{\,\xrightarrow{\widthof{\,#1\,}}\,}}}
\newcommand{\convweak}[1]%
  {{\mathrel{\,\xrightarrow{\widthof{\,#1\text{-w.}\,}}\,}}}

\newcommand{\ubar}[1]{\text{\b{$#1$}}}
\DeclareRobustCommand{\munderbar}[1]{\underaccent{\bar}{#1}}
\renewcommand{\qedsymbol}{$\Box$}

\newtheoremstyle{customtheorem}% name of the style to be used
  {0.5em}% measure of space to leave above the theorem. E.g.: 3pt
  {0.2em}% measure of space to leave below the theorem. E.g.: 3pt
  {\itshape}% name of font to use in the body of the theorem
  {}% measure of space to indent
  {\scshape}% name of head font
  {}% punctuation between head and body
  {1ex}% space after theorem head; " " = normal interword space
  {}% Manually specify head

\theoremstyle{customtheorem}
\newtheorem{theorem}{Theorem}[section]
\newtheorem{lemma}[theorem]{Lemma}
\newtheorem{proposition}[theorem]{Proposition}

\newtheorem{definition}[theorem]{Definition}

\newtheoremstyle{customremark}% name of the style to be used
  {0.5em}% measure of space to leave above the theorem. E.g.: 3pt
  {0.2em}% measure of space to leave below the theorem. E.g.: 3pt
  {}% name of font to use in the body of the theorem
  {}% measure of space to indent
  {\scshape}% name of head font
  {}% punctuation between head and body
  {1ex}% space after theorem head; " " = normal interword space
  {}% Manually specify head

\theoremstyle{customremark}
\renewenvironment{proof}{\par\noindent{\scshape Proof}\;}{\hfill\qedsymbol\par}
\newtheorem{remark}[theorem]{Remark}

\begin{document}

\thispagestyle{empty}

\title{
Clearing price distributions in call auctions
 }
\author{
  M. Derksen\thanks{Corresponding author.
Email: m.j.m.derksen@uva.nl} $\dag \ddag$, B. Kleijn$\ddag$ and R. de Vilder$ \dag \ddag$ \\[1mm]
  {\small\it  $\dag$ Deep Blue Capital N.V., Amsterdam}\\
  {\small\it  $\ddag$ Korteweg-de~Vries Institute for Mathematics,
    University of Amsterdam}
  }
\date{\today}
\maketitle

\begin{abstract}\noindent
We propose a model for price formation in financial markets
based on clearing of a standard call auction with
random orders, and verify its validity for prediction of the
daily closing price distribution statistically.
The model considers random buy and sell orders, placed
following demand- and supply-side valuation distributions; an
equilibrium equation then leads to a distribution for clearing
price and transacted volume. Bid and ask
volumes are left as free parameters, permitting possibly
heavy-tailed or very skewed order flow conditions. In
highly liquid auctions, the clearing price distribution
converges to an asymptotically normal central limit, with
mean and variance in terms of supply/demand-valuation
distributions and order flow imbalance. By means of simulations, we
illustrate the influence of variations in order flow and
valuation distributions on price/volume, noting a distinction between
high- and low-volume auction price variance. To verify the validity of
the model statistically, we predict a year's worth of daily
closing price distributions for 5 constituents of the Eurostoxx 50
index; Kolmogorov-Smirnov statistics and QQ-plots demonstrate
with ample statistical significance that the model predicts
closing price distributions accurately, and compares
favourably with alternative methods of prediction.
\end{abstract}

%%%%%%%%%%%%%%%%%%%%%%%%%%%%%%%%%%%%%%%%%%%%%%%%%%%%%%%%%%%%%%%%%%%%%%%%%%%%%%%

\section{Introduction}
\label{sec:intro}
In modern financial markets most securities are traded in
continuous double auctions. During the trading day a
sell/buy order for a price lower/higher than or equal to the
best bid/ask price is immediately executed versus the
limit order book on the bid/ask side. If a sell/buy-order has a price
higher/lower than the best bid/ask, it is added to the
limit order book on the ask/bid side. To start and stop trading and determine
daily opening and closing prices, standard call auctions are
conducted for most securities. In these opening and closing auctions buy and sell
orders are collected over a set interval in time, after which a
\emph{clearing price} $X$ is determined to clear the maximal
executable volume \cite{EN}, transacting all against the price $X$.

A large part of the market microstructure literature focusses on detailed modelling of
continuous double auctions and the limit order book. There are essentially two different lines of work: equilibrium models in which order arrival is governed by decisions of individual agents trying to maximize utility (see \eg\ \cite{Parlour98,Foucault99,Goettleretal05,Rosu09,Bressanfacchi13,Bressanwei16} ) and stochastic limit order book models in which order arrival is completely stochastic 
(see \eg, among many others, \cite{Luckock03,Smithetal03,Contetal10,Abergeljedidi13,Contdelarrard13,Toke15_2}).
Some extensive
studies of empirical properties of the limit order book can be found in \cite{Biaisetal95,ChalletStinchcombe01,Bouchaudetal02,PottersBouchaud03}.
The standard call auction has received less attention: 
\cite{Mendelson82} models a call auction in which all
orders have size one and are uniformly distributed over some
price interval, while buy and sell orders arrive \cf\ a homogeneous
Poisson process. The distribution of transacted volume is derived,
together with the clearing price expectation. Technically, this
paper is related to the work of \cite{Toke15}, who gives
the full solution of Mendelson's call auction model,
deriving distributions for transacted volume, and lower/upper
clearing prices, as well as asymptotic distributions in very
liquid call auctions.

At the conceptual level, our approach is related to the
seminal paper by \cite{Smithetal03}, who consider a
statistical model for continuous double auctions assuming
\iid\ random order flow, modelled through independent, homogeneous
Poisson processes for market orders, limit orders and cancellations
with random order-prices from a single, uniform valuation distribution.
Simulations, dimensional analysis and mean-field approximations then
lead to predictions for price volatility, market depth, price-impact
function, bid-ask spread and probability/time to fill a limit order.

In this paper we propose a model for price formation
in financial markets with a \emph{bid/ask equilibrium equation} at
its core, that sets the \emph{clearing price} such as to lead to
maximal transacted volume, based on fixed numbers $N_A,N_B$ of
unit-sized sell and buy orders forming \iid\ samples from
distinct \emph{valuation distributions} $F_A$ and $F_B$. Due to the
randomness in the orders, the equilibrium gives rise to a
distribution for $X|N_A,N_B$, the clearing price conditional
on $N_A,N_B$. The shape of the valuation distributions $F_A,F_B$ and
the distribution of  the pair $(N_A,N_B)$ remain unspecified;
while the former models order density, the latter permits
great freedom of modelling order flow conditions, including
auctions in which extreme or skewed liquidity-conditions disturb
equilibria and distort clearing prices. We derive closed-form
expressions for distributions of clearing prices, jointly with
transacted volumes.

Such mechanisms have direct application in the modelling of
opening and closing auctions as demonstrated with data from
intraday transactions to predict closing price distributions of several constituents of the Eurostoxx 50 index (roughly speaking, this index consists of the 50 main Eurozone companies) in
section~\ref{sec:application}. Extending the argument more
informally, we argue that the model applies also in
\emph{continuous} trading: if
buy/sell orders are accrued over a period of time (and liquidity
providers trade with a more-or-less neutral combined inventory)
then, at the aggregate level, the detailed
process of trading during the period can be interpreted
as market-clearing at a price $X$ with a distribution
that depends on valuation distributions $F_A,F_B$ and the
distribution of the pair $(N_A,N_B)$ that reflects order flow
conditions during the interval. If liquidity providers do not trade
neutrally, or if we take a limit order book into account, the
equilibrium between newly accrued buy and sell orders is
perturbed by so-called \emph{excess liquidity}, which can be taken
into account in full generality and lies at the heart of many
interesting properties associated with real-world phenomena.

The remainder of this article is structured as follows. In section \ref{sec:model}, the model is introduced, probability distributions for clearing price and volume are derived and several proposals for the order flow distributions are made. In section~\ref{sec:asymp}, we consider auctions in which
the number of incoming orders is very
large. Asymptotically the
clearing price has a normal distribution, which implies that if we
approximate continuous trading by a periodically cleared market,
the resulting discrete price process follows a Brownian
path. This is roughly in support of general pricing models
based on the efficient market hypothesis, with mean and variance
of the return distribution expressed in terms of the distributions
of supply, demand and order flow.  In section~\ref{sec:supdem} we explore how changing supply and demand distributions affect the joint distribution of clearing price and transacted volume, leading to a distinction between two different types of auction price variance; one occuring when transacted volumes are high, the other one when these are low.
In section ~\ref{sec:impact} we study the model's perspective on the \emph{price impact} of market orders. Remarkably, the model
reproduces the concave price impact functions observed empirically
\cite{Hasbrouck91,Lilloetal03,Donierbonart15} and explained theoretically
\cite{Smithetal03,Donieretal15,Benzaquenbouchaud18}. In section~\ref{sec:application} the model
is applied to estimate the distribution of the clearing price of
a closing auction, based on the day's transaction data.
For 5 (randomly selected) constituents of the Eurostoxx 50 index, it is
shown that the model predicts the probability distribution of the
closing price with precision, through assessment of QQ-plots
and Kolgomorov-Smirnov statistics. For comparison, a more crude
alternative method of estimation is assessed on the same basis.
It is shown that the market clearing model provides significantly
better estimates for clearing price distributions than this more
straightforward method. Most important results are summarized
in the concluding section~\ref{sec:conclusion}. Proofs of the
theoretical results of sections~\ref{sec:model} and~\ref{sec:asymp}
as well as notation and conventions are collected in
appendix~\ref{app:proofs}.

%%%%%%%%%%%%%%%%%%%%%%%%%%%%%%%%%%%%%%%%%%%%%%%%%%%%%%%%%%%%%%%%%%%%%%%%%%%%%%%

\section{Stochastic market clearing}
\label{sec:model}

In this section, we introduce the model and derive expressions for the
distributions of central quantities in the clearing process.

%%%%%%%%%%%%%%%%%%%%%%%%%%%%%%%%%%%%%%%%%%%%%%%%%%%%%%%%%%%%%%%%%%%%%%%%%%%%%%%%

\subsection{Supply/demand equilibrium}

Let us consider a standard call auction for a given asset. In the auction,
buy and sell orders are matched to transact at a clearing price $X$,
determined in such a way that the total transacted volume is maximal.
Suppose that $N_A$ sell orders are submitted, as well as $N_B$ buy
orders and that every order has equal size (set to one). We assume that
participants on both sides of the market formulate their orders
independently of each other, according to certain \emph{valuation distributions}.
That is, we model the ask prices as an \iid\ sample $(A_1, \dots A_{N_A})$
from a \emph{supply (or ask) distribution} $F_A$ and the bid prices as
an \iid\ sample $(B_1, \dots, B_{N_B})$ from a \emph{demand (or bid)
distribution} $F_B$. The interpretation of $F_A$ is as follows: the
probability that a randomly selected seller is
willing to sell the asset for an ask price $A \leq x$, is given by $F_A(x)$, for
all $x \in \RR$. Similarly, if we randomly select a buyer, the
probability that he is willing to buy the asset for a bid price $B \leq x$ is
given by $F_B(x)$. Naturally ask prices are higher than bid
prices, however, the ordering is expressed through the supply and demand 
distributions $F_A,F_B$, through the assumption that,
\begin{equation}
  \label{eq:stochorder}
  F_A(x)=\PP(A\leq x)\leq P(B\leq x)=F_B(x)
\end{equation}
for all $x\in\calX$. Note that the general ordering of buy prices below sell prices
cannot be defined in very strict or deterministic ways; the 
uncertainty in $A,B$ and the stochastic nature of the ordering
enables crossing prices and thereby, matchable orders and the auction
itself. For reasons of technical feasibility, it is assumed that buyers and sellers formulate their quotes independently, \ie\ bid-
and ask-samples are independent \iid\ samples. \footnote{Of course these independence assumptions are not realistic: especially when prices fluctuate a lot, it is likely that market participants on both sides of the market react on each other's decisions and hence their quotes are far from independent. However, we argue that, despite these simplifying assumptions, the model can still be interpreted as a reasonable description of price formation in auctions, as is confirmed by the results in Section \ref{sec:application}.}
Denote by $\FB$ and $\FA$ the empirical distribution functions associated
with the bid- and ask-samples $(B_1, \dots,B_{N_B})$ and
$(A_1,\dots,A_{N_A})$, that is,
\[
  \FB(x) = \frac{1}{N_B}\sum_{j=1}^{N_B} 1_{\{B_j \leq x \}},
  \quad
  \FA(x) = \frac{1}{N_A} \sum_{i=1}^{N_A} 1_{\{A_i \leq x\}}.
\]
For every $x \in \RR$, denote the number of submitted sell orders with a
price less than or equal to $x$ by $\DA(x)$ and the number of submitted
buy orders with a price greater than $x$ by $\DB(x)$.
As discussed above, the \emph{clearing price} $X$ is obtained by maximizing the
total transacted volume. In terms of the above defined quantities, that
implies $X$ is defined as a solution of the \emph{market clearing equation}
$\DA(X) = \DB(X)$, or,
\begin{equation}
\label{eq:eq_equation}
  N_A \FA(X) = N_B\bigl(1-\FB(X)\bigr), 
\end{equation}
which expresses that the transacted volume is maximized at (any) price $X$
where the supply curve $\DA$ and the demand curve $\DB$ intersect.
Consider the following definition.
\begin{definition}
\label{def:eqprice}
For a given sell order sample $(A_1,\dots A_{N_A})$ from $F_A$ and a
buy order sample $(B_1,\dots,B_{N_B})$ from $F_B$, the corresponding
\emph{clearing price} $X$ is defined by 
\[
  X=\inf\{x \in \RR : \DA(x) \geq \DB(x)\}.
\]
\end{definition}
\begin{remark}\label{rem:eqpricedef}
It should be noted that there are issues of existence and uniqueness
of solutions to (\ref{eq:eq_equation}). Firstly, when the bid- and
ask-samples are such that,
\[
  B_1 \leq \dots \leq B_{N_B} < A_1 \leq \cdots \leq A_{N_A}
\]
there is \emph{no solution} where $\DA$ and $\DB$ intersect. Secondly,
it is possible that there is an interval $[\ubar{X},\bar{X}]$ of
possible clearing prices for which $\DA=\DB$, ruining uniqueness.
Both issues are addressed in definition~\ref{def:eqprice}, much in
the same way quantiles of a distribution are defined
(see figure \ref{fig:da_db_simple} for an illustration).
\begin{figure}
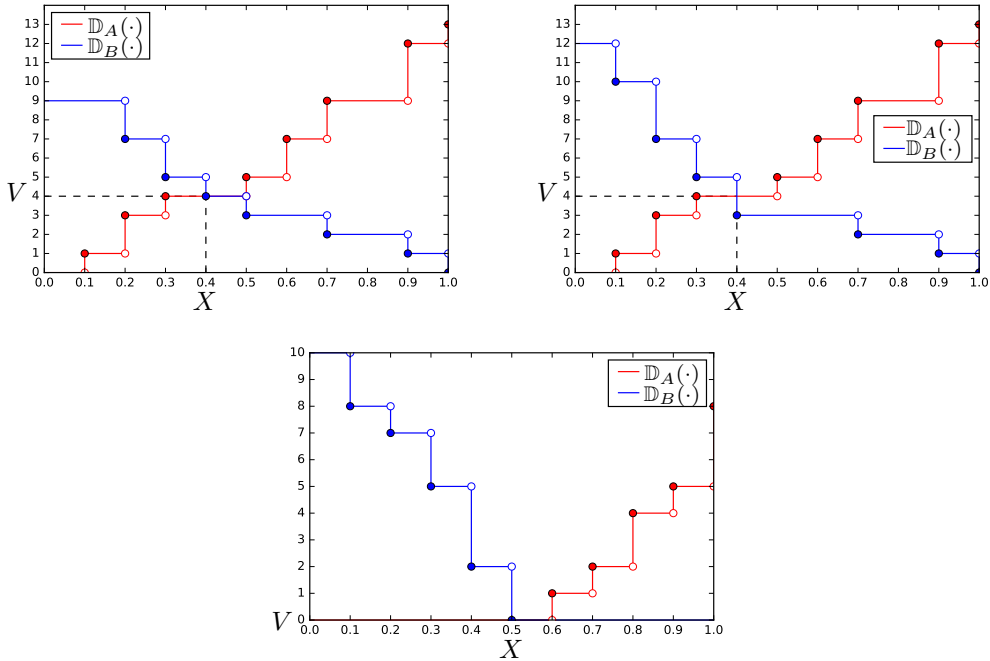

  \begin{center}
    \begin{lpic}{da_db_simple_v2(0.45)}
      \lbl[t]{40,88;{\scriptsize{$\DA(\cdot)$}}}
      \lbl[t]{40,82;{\scriptsize{$\DB(\cdot)$}}}      
      \lbl[t]{66,7;{\small{$X$}}}
      \lbl[t]{11,39;{\small{$V$}}}
    \end{lpic}
    \begin{lpic}{da_db_cross_v2(0.45)}
      \lbl[t]{124,58;{\scriptsize{$\DA(\cdot)$}}}
      \lbl[t]{124,52;{\scriptsize{$\DB(\cdot)$}}}   
      \lbl[t]{66,7;{\small{$X$}}}
      \lbl[t]{11,39;{\small{$V$}}}
    \end{lpic}
    \begin{lpic}{da_db_disjoint_v2(0.45)}
      \lbl[t]{125,88;{\scriptsize{$\DA(\cdot)$}}}
      \lbl[t]{125,82;{\scriptsize{$\DB(\cdot)$}}}     
      \lbl[t]{78,7;{\small{$X$}}}
      \lbl[t]{11,15;{\small{$V$}}}
    \end{lpic}
    \caption{
      \label{fig:da_db_simple}
      Three possible examples of the supply curve $\DA(\cdot)$ (the increasing
      (red) step function) and the demand curve $\DB(\cdot)$ (the decreasing
      (blue) step function).
      Left upper panel: a situation in which there is no unique point of
      intersection, note the position of $X$ at the left of the interval where
      $\DA=\DB$. Right upper panel: a situation in which there is a unique
      intersection point, but $\DA(X)> \DB(X)$. Lower panel: A situation in
      which no transactions are possible, note the position of $X$ at the
      highest placed buy order. Note also the position of the transacted volume ($V$) after clearing (this quantity is defined later on).}
  \end{center}
\end{figure}
\end{remark}
In subsequent subsections, closed-form expressions are provided for the
probability distributions (conditional, given $(N_A,N_B)$) of several
important market quantities, like clearing price $X$ and transacted
volume $V$. 

While this stochastic model of price formation is based on the
mechanism of a call auction, the clearing price also
has an interpretation for continuous trading. To appreciate the
relation, the process of continuous bidding and transacting
(with matching of orders as an instantaneous but momentary form
of clearing) should be viewed in an aggregated form over an
interval of time $I$. During any such interval the numbers of buyers
and sellers must still be equal, and that is exactly what equation
(\ref{eq:eq_equation}) expresses. Then, at the aggregate level, the
detailed, step-by-step process of trading during the interval
may be modelled equivalently (or in close approximation) as market
clearing at a clearing price $X$ associated with the interval $I$.

For both the auction and the continuous trading interpretations, the
following applies: if $F_A$, $F_B$ and the distribution of $(N_A,N_B)$
are chosen in an appropriate way, the clearing price $X$ can be interpreted
as a true, underlying price for the asset, associated not with any
specific point in time but with the whole interval $I$ (to relate such
an \emph{interval-price} to timed market prices, one may think of $X$
loosely as the price \emph{at a time $T$ randomly sampled from $I$}).
To justify the fixed distributions $F_A$, $F_B$ and the independence
assumptions on the order samples, $I$ must not be too long due to
possible non-stationarity but long enough statistically, aggregating a
sufficiently large numbers of orders. Furthermore, the stochastic
behaviour of $\FA$ and $\FB$ (that is, the randomness these quantities
represent) must reflect the uncertainty in the incoming orders on
the respective sides of the market during the time interval $I$ with
some accuracy. Similarly the distribution chosen for liquidity
$(N_A,N_B)$ must reflect the uncertainty in actual market liquidity
conditions during the interval $I$. If these conditions are met, the
model will provide an accurate reflection of the stochastic aspects of
market clearing, and thereby, of price formation.

In the setting of continuous trading, it makes sense to measure time
in terms of market events rather than physical time, in particular
regarding the interval $I$. Combining with the interpretation of $X$
as a true, underlying price for the interval $I$, we can fix $N=N_A+N_B$
and interpret the resulting clearing price $X$ as a true, underlying
asset price associated with the interval spanned by the next $N$ orders.

The unrestricted freedom in the choices for $F_A,F_B$ and the distribution
of $(N_A,N_B)$ enables use of empirical fits for these distributions from
previous intervals. It is also possible to make definite, default choices
for these quantities: for instance, choosing independent Poisson
distributions for $N_A$ and $N_B$ would correspond to the assumption of
Poisson order flow, which is omnipresent in the literature (see, among
many others, \cite{Smithetal03,Contetal10,Abergeljedidi13,Contdelarrard13, Toke15_2}
for examples in context of continuous double auctions and
\cite{Mendelson82,Toke15} for examples in the standard call auction).
In subsection \ref{sec:liqdist} we consider further possible choices for
the distribution of $(N_A,N_B)$ and the model properties implied.

%%%%%%%%%%%%%%%%%%%%%%%%%%%%%%%%%%%%%%%%%%%%%%%%%%%%%%%%%%%%%%%%%%%%%%%%%%%%%%%%

\subsection{Distribution of clearing price and volume}
In this subsection we derive the probability distributions of
price and price-volume, resulting from the equilibrium
equation~(\ref{eq:eq_equation}), without and with a limit
order book. We concentrate on the marginal distribution of the clearing
price $X$ only first, given by the following theorem (proved in the appendix).
\begin{restatable}[Clearing price distribution]{theorem}{thmcpdist}
\label{thm:cpdist}
The distribution of the clearing price $X$, conditional on $N_A$ and $N_B$,
is given by,
\[
  \begin{split}
  \PP(X &\leq x | N_A, N_B)\\
    &=\sum_{k=0}^{N_A} \sum_{l=0}^{N_B \wedge k} \binom{N_A}{k}
    F_A(x)^k(1-F_A(x))^{N_A-k} \binom{N_B}{l} (1-F_B(x))^l F_B(x)^{N_B-l}.
  \end{split}
\]
\end{restatable}
However, it is also possible to derive the joint distribution
of clearing price and transacted volume, which is defined next.
\begin{definition}
\label{def:eq_volume}
The transacted volume $V$ corresponding to the clearing price $X$, is defined by $V=\DA(X)$. 
\end{definition}
\begin{remark}
\label{rem:Vambiguity}
The quantity $V=\DA(X)$ should be interpreted as the maximal number of orders that can be matched in clearing. In the context of a call auction, it is the total volume that is transacted. If $F_A$ and$F_B$ are continuous distributions, there is almost surely a unique point where $\DA$ and $\DB$ intersect, hence $V=\DA(X)=\DB(X)$. In the case of a discrete price-axis it is possible that $\DA(X)>\DB(X)$, which means that the volume $\DA(X)$ is not completely matched (see the upper right panel of figure~\ref{fig:da_db_simple}). As a convention, we
neglect such discretization effects and continue with
definition~\ref{def:eq_volume} (compare with the resolution to the
ambiguity for $X$, as an infimum, see remark~\ref{rem:eqpricedef}).
\end{remark}
In the next theorem (proved in the appendix), an explicit expression for
the joint distribution of $X$ and $V$ is provided. It is assumed that
the price-axis $\calX$ is a discrete set,
$\calX :=\{x_0, x_0 + \delta, \dots\}$, where $\delta$ is the \emph{ticksize}. 
\begin{restatable}[Joint clearing price/transacted volume distribution]%
{theorem}{thmxv}
\label{thm:xvdist} 
The joint distribution of clearing price $X$ and transacted volume $V$,
conditional on $N_A$ and $N_B$, is given by,
\begin{equation}
\label{eq: result_xvdist}
\begin{split}
  \PP(&X\leq x, V \leq v|N_A,N_B)\phantom{\sum_{u=0}^v}\\
  & = \sum_{u=0}^v \sum_{k=0} ^u \sum_{l=0}^k
    \biggl[ \binom{N_B}{l}\binom{N_A}{k,u-k,N_A-u}(1-F_B(x))^l F_B(x)^{N_B-l}\\
  &\qquad \qquad\times F_A(x)^k(F_A(x+\delta)-F_A(x))^{u-k}(1-F_A(x+\delta))^{N_A-u}
    \biggr]\phantom{\sum_{u=0}^v}\\
  &\qquad + \sum_{y \in \mathcal{X}, y\leq x} \sum_{u=v+1}^{N_A}
    \sum_{k=0}^u \sum_{l=0}^k
    \biggl[ \binom{N_B}{l}\binom{N_A}{k,u-k,N_A-u}(1-F_B(y))^l F_B(y)^{N_B-l}\\
  &\qquad \qquad \times F_A(y)^k(F_A(y+\delta)-F_A(y))^{u-k}(1-F_A(y+\delta))^{N_A-u}
    \biggr]\phantom{\sum_{u=0}^v}\\
  &\qquad - \sum_{y \leq x} \sum_{l=0}^{N_B} \sum_{k= l \vee v+1}^{N_A}
    \binom{N_A}{k} F_A(y)^k (1-F_A(y))^{N_A-k}
    \binom{N_B}{l} F_B(y)^{N_B-l} (1-F_B(y))^l.
  \end{split}
\end{equation}
\end{restatable}

\subsection{Excess liquidity}

There are several variations possible on the definition of the clearing
price $X$ as given above: to start with, during continuous trading,
exchanges often offer an open limit order book, which contains all
visible limit orders on ask-side and bid-side. Denote by $L_A(x)$ the
total volume on the ask-side of the limit order book for a
price less than or equal to $x$. Similarly, denote by $L_B(x)$ the
total volume on the bid-side of the limit order book
for a price above $x$. Then definition~\ref{def:eqprice} of
the clearing price $X$ is adapted to,
\[
  X=\inf\{x \in \RR : \DA(x)+ L_A(x) \geq \DB(x)+L_B(x)\},
\]
corresponding to an adapted market clearing equation that takes the limit order book into account:
\begin{equation}
\label{eq:limitbook_eq_equation}
  N_A \FA(X) + L_A(X) = N_B\bigl(1-\FB(X)\bigr) + L_B(X), 
\end{equation}
Note that $x\mapsto L_A(x)$ and $x\mapsto L_B(x)$ are non-stochastic
quantities and that for any $x$, either $L_A(x)$ or $L_B(x)$ is equal
to zero (as, otherwise, the book could be cleared further by matching the
overlapping orders).

To generalize, we include \emph{excess liquidity} as any sort of liquidity
that plays a role in the clearing process, but does not originate from the
quoting process as described by $\FA$ and $\FB$. As such, we view excess
liquidity as an external influence.
\begin{definition}
\label{def:refprice_el}
If the clearing price $X$ is defined by the equation,
\begin{equation}
  \label{eq:refprice_el}
  N_A \FA(X) = N_B\bigl(1-\FB(X)\bigr) + \Delta(X),
\end{equation}
where $\Delta: \mathcal{X} \to \ZZ$ is a right-continuous,
non-increasing function, then $\Delta$ is called the \emph{excess liquidity}. 
\end{definition}
Excess liquidity takes the market out of the `pure' equilibrium given by
$\DA(X) = \DB(X)$. For example, inclusion of the limit order book is
possible through $\Delta(x)=L_B(x) - L_A(x)$. Positive values of
$\Delta(x)$ correspond to an \emph{excess demand} and negative values
of $\Delta(x)$ mean an \emph{excess supply}. Another example of excess
liquidity is the arrival of a market order. A sell market order
of size $\omega \in \NN$ corresponds to the constant function
$\Delta = -\omega 1_\calX$, while a buy market order is described by
the function $\Delta = \omega 1_\calX$. Similarly, a buy limit
order with limit price $b$ can be described by $\Delta = \omega 1_{[x_0,b]}$
and a sell limit order with limit price $a$ by
$\Delta = -\omega 1_{[a,\infty)}$. 

Lemma \ref{lem:technlemma} can be re-derived with excess liquidity, in
order to obtain the equivalence
$X \leq x \Leftrightarrow \DA(x) \geq \DB(x) + \Delta(x)$.
Exactly like in the proof of theorem \ref{thm:cpdist}, this leads to the distribution of the clearing price, conditional on $N_A$ and $N_B$, as stated in the next proposition.
\begin{proposition} [Clearing price distribution in case of excess liquidity] 
\label{prop:eqprice_el_dist}
When excess liquidity $x\mapsto\Delta(x)$ plays a role, the clearing price
distribution conditional on $N_A,N_B$, is given by 
\[
  \begin{split}
  \PP(X \leq x &| N_A, N_B) \\
    &= \sum_{k=0}^{N_A} \sum_{l=0}^{U(k,x)}
      \binom{N_A}{k} F_A(x)^k(1-F_A(x))^{N_A-k}
      \binom{N_B}{l} (1-F_B(x))^l F_B(x)^{N_B-l},
  \end{split}
\]
where $U(k,x) = (k -\Delta(x)) \wedge N_{B}$.
\end{proposition}
Note that the limit order book makes an appearance \emph{only} in the summation
bound, leaving the binomial character of the equilibrium distribution intact.

\subsection{Order flow distributions}
\label{sec:liqdist}
All the distributions derived in the previous subsections, are conditional on
the pair $(N_A,N_B)$. In this subsection we discuss some possibilities
for the distribution of $(N_A,N_B)$ (the so-called \emph{order flow
distribution}) and their consequences for clearing price distributions.
The common assumption in the (early) literature is what is called \emph{Poisson order flow}:
for continuous double auctions  \cite{Smithetal03,Contetal10,Abergeljedidi13,Contdelarrard13,Toke15_2} and call auctions \cite{Mendelson82, Toke15},
Poisson order flow follows from assumed independent Poisson processes
for the arrival of buy and sell orders. Here, we would take,
\[
  (N_A,N_B) \sim \text{Pois}(\mu_A T) \times \text{Pois}(\mu_B T),
\]
for Poisson rates $\mu_A$, $\mu_B$ and a given interval duration $T$
to achieve the same.

However, in this setting it makes sense to consider more general
models for order flow. Assume again that we consider an interval in which
$N$ new orders arrive. Fix $N_A+N_B=:N \in \NN$ and leave the distribution
of $N_A$ open for choice. A reasonable choice would be to choose $N_A$
according to \emph{binomial order flow}, \ie
\[
  N_A \sim \Bin(N,p),
\]
for some $p \in (0,1)$ representing order flow imbalance. Taking,
\[
  (N_A,N_B) \sim  \text{Pois}(\mu_A) \times \text{Pois}(\mu_B),
\]
is equivalent to,
\[
  N = N_A +N_B \sim \text{Pois}(\mu_A+\mu_B), \quad N_A |N \sim \Bin(N,p),
\] 
for $p=\mu_A/(\mu_A+\mu_B)$.

Both Poisson and binomial proposals express
the conviction that \emph{order flow imbalance} $\alpha:=N_A/N$ does not
display great stochastic fluctuation and lies close to its expectation $p$,
especially for greater values of $N$ due to the central limit theorem.
This makes it difficult to capture market phenomena that are due to fat
tails in the order flow distribution, to describe more extreme, yet common
market conditions. Hence our third proposal: we consider \emph{beta order flow imbalance},
\[
  N_A = \alpha N, \quad \alpha \sim \text{Beta}(\beta_1,\beta_2).
\]
Choice of the parameters $\beta_1,~\beta_2$ permits great modelling
freedom. For instance, if we expect the order flow on the bid- and ask-side
to be roughly balanced, it is appropriate to set $\beta_1=\beta_2$. If we
expect the market to be out of balance (\eg\ while trending), we may
choose $\beta_1>\beta_2$ when we expect more supply than demand, and
vice versa. Perhaps most interesting is the scale of the betas: if
$\beta_1,\beta_2<1$ we induce the fat tails not seen in Poisson or
binomial order flow, while $\beta_1,\beta_2\gg 1$ will lower the variance
and bring $\alpha$ close to its expectation $\beta_1/(\beta_1+\beta_2)$.

To shed more light on the influence of the order flow distribution on
the clearing price distribution, we consider a simple example. To focus
on order flow, we make the trivial choices for the other parameters:
$F_A(\cdot)=F_B(\cdot)=\Phi_{\mu,\sigma}(\cdot)$ for
$\mu =10$ and $\sigma =0.1$. To appreciate the effects of order flow on clearing price distributions, consider figure~\ref{fig:role_liqdist_balanced}, the probability
density $f_X$ of the clearing price is plotted for various balanced (left panel) and unbalanced (right panel) choices
of the order flow distribution. As expected, $f_X$ centers around $10$ in
all balanced cases and around a lower location for the unbalanced cases. It is seen that
Poisson order flow leaves little room for variation in the values of $N_A$
and $N_B$, causing the density to peak relatively sharply. By contrast, beta
order flow imbalance leads to more liquidity-driven uncertainty in the clearing
price.
\begin{figure}[h]
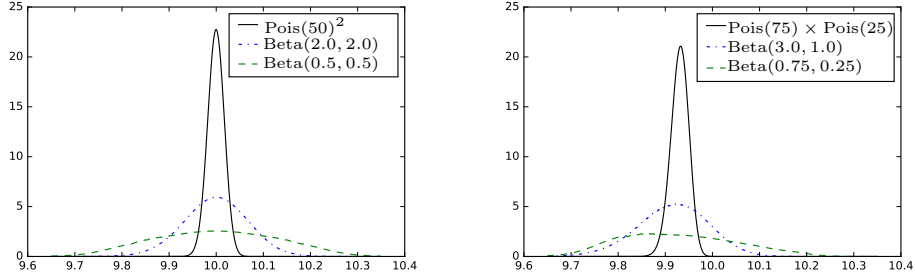

  \centering
    \parbox{13.5cm}{
    \centering
    \begin{lpic}{role_liqdist_balanced_v2(0.42)}
      \lbl[t]{106,88;{\tiny{$\text{Pois}(50)^2$}}}
      \lbl[t]{111,82;{\tiny{$\text{Beta}(2.0,2.0)$}}}
      \lbl[t]{111,76;{\tiny{$\text{Beta}(0.5,0.5)$}}}
    \end{lpic}
      \begin{lpic}{role_liqdist_unbalanced_v2(0.42)}
      \lbl[t]{109,87;{\tiny{$\text{Pois}(75)\times\text{Pois}(25)$}}}
      \lbl[t]{101,81;{\tiny{$\text{Beta}(3.0,1.0)$}}}
      \lbl[t]{104,75;{\tiny{$\text{Beta}(0.75,0.25)$}}}
    \end{lpic}
    \caption{\label{fig:role_liqdist_balanced} Clearing
    price density $f_X$, for $F_A=F_B=\Phi_{10,0.1}$ and various
    balanced (left panel) and unbalanced (right panel) choices for the order flow distribution. The solid black line
    is the density for Poisson order flow, while the dashed blue line
    and the dashed green line correspond to beta order flow imbalance, for fixed $N=100$.}}
\end{figure}

%%%%%%%%%%%%%%%%%%%%%%%%%%%%%%%%%%%%%%%%%%%%%%%%%%%%%%%%%%%%%%%%%%%%%%%%%%%%%%%%
%%%%%%%%%%%%%%%%%%%%%%%%%%%%%%%%%%%%%%%%%%%%%%%%%%%%%%%%%%%%%%%%%%%%%%%%%%%%%%%%

%%%%%%%%%%%%%%%%%%%%%%%%%%%%%%%%%%%%%%%%%%%%%%%%%%%%%%%%%%%%%%%%%%%%%%%%%%%%%%%%

\section{The high-liquidity limit}
\label{sec:asymp}

In this section, we provide the asymptotic clearing price distribution in limit
of infinite liquidity. To be more precise, denote $N=N_A+N_B$, let 
$N_A = \alpha N$, $N_B=(1-\alpha)N$ for some constant $0<\alpha <1$
we refer to as \emph{order flow imbalance} and consider the
limit $N \to \infty$.
We take a continuous price-axis $\calX = [x_0,\infty)$ and assume that the
distribution functions $F_A$ and $F_B$ are strictly increasing,
describing measures that are absolutely continuous with respect
to the Lebesgue measure, with densities denoted $f_A$ and $f_B$.
Let $X$ denote a solution to $N_A \FA(X) = N_B(1-\FB(X))+\Delta(X)$, with
possibly non-zero excess liquidity $\Delta$. Denote by $x_E$ the
\emph{real equilibrium price} which is the (non-random) price uniquely
defined by the equilibrium equation,
\begin{align}
\label{eq:real_eq_price}
   \alpha F_A(x_E) = (1-\alpha)(1-F_B(x_E)).
\end{align}
According to the following theorem (the proof of which can be found in the
appendix), the clearing price $X$ is in the limit
distributed according to a normal distribution centred on $x_E$ with variance
that depends on $f_A$ and $f_B$.
\begin{restatable}[High-liquidity clearing price distribution]{theorem}{highliqlimit}
\label{thm:highliq}
Let $X$ be the clearing price in case of possible
excess liquidity $\Delta$. Assume that $F_A$ and $F_B$ are
strictly increasing and absolutely continuous with respect to the
Lebesgue measure with densities $f_A$ and $f_B$. Additionally, assume
that excess liquidity scales with $N$ as $\Delta(\cdot) = \sqrt{N} D(\cdot)$,
for some continuous and bounded function $D : \calX \to \RR$. Then,
as $N \to \infty$,
\begin{equation}
  \label{eq:highliq}
  \sqrt{N}(X-x_E) \conv{w.} N(\mu(x_E),\sigma^2(x_E)),
\end{equation}
where the asymptotic mean and standard deviation are given by,
\[
  \mu(x_E)=\frac{D(x_E)}{\alpha f_A(x_E)+(1-\alpha)f_B(x_E)},
  \quad
  \sigma(x_E) = 
  \frac{\tau(x_E)}{\alpha f_A(x_E) + (1-\alpha)f_B(x_E)},
\]
for
\[
  \tau^2(x_E) = \alpha F_A(x_E)\bigl(1-F_A(x_E)\bigr)
  + (1-\alpha) F_B(x_E)\bigl(1-F_B(x_E)\bigr),
\]
and $x_E$ is the real equilibrium price.
\end{restatable}
Consider a standard call auction in which the number of orders collected
is very large. The clearing price distribution is then closely concentrated
around $x_E$ and has
a width proportional to $1/\sqrt{N}$. So the model confirms
the intuition that large auctions lead to accurate price discovery and adds
that this accuracy is inversely proportional to the square root of the number
of orders. Non-zero excess liquidity of order $\sqrt{N}$ biases $X$ away
from $x_E$, however, this bias is also proportional to $1/\sqrt{N}$. So the
model says that in highly liquid auctions or markets, external influence in
the form of excess liquidity $\Delta$ must be of order larger than $\sqrt{N}$
to force (the distribution of) the clearing price away from the real
equilibrium price $x_E$. 
Furthermore,
the shift caused by the excess liquidity is inversely proportional
to a convex combination of $f_A$ and $f_B$, hence price impact will be
larger if the density of orders around the equilibrium price is low.

Next consider the case of continuous trading of a stock in an interval,
during which supply and demand are described by the distributions $F_A$
and $F_B$, and by order flow imbalance $\alpha \in (0,1)$. Assume that the
number of incoming orders during the interval
is very large, so that the limit of theorem~\ref{thm:highliq} forms a
good approximation for the clearing price distribution.
In the absence of excess liquidity, the distribution of the clearing price
associated with the interval is a sharply peaked normal distribution centred
at the real equilibrium price. If we repeat this argument for consecutive
intervals (possibly with changing $F_A$, $F_B$ and $\alpha$) and approximate
continuous trading by a periodically cleared market, the price process
becomes a discrete Brownian path (possibly trending if we add excess
liquidity). In many stochastic models for pricing, this type of stochastic
process is postulated; by contrast, here, the Brownian path emerges from
the central limit (in the form of Donsker's theorem, see the proof of
theorem~\ref{thm:highliq}) and the parameters of this Brownian path have
an interpretation in terms of supply, demand and order flow imbalance.

As argued after definition~\ref{def:refprice_el}, the model invites the
interpretation of the limit book as excess liquidity, in a
market made around the equilibrium price $x_E$ of many new orders. Think,
for instance, of a situation where new orders originate from liquidity
providers primarily; if the location of their equilibrium price
distribution undergoes a small but quick jump (for example because of
a sudden change in the price of a hedging index future), the result of theorem \ref{thm:highliq} suggests that the limit book
obstructs immediate market correction: if we consider a limit book
of order greater than $N^{1/2}$ over an interval of order $N$ orders,
the location of the clearing price distribution is expected to differ
from the liquidity providers' new $x_E$ on scales larger than $N^{-1/2}$.  To
re-phrase that slightly and more crudely, the model suggests that a limit book offering total
liquidity of order $L$ stabilizes the market price versus
fluctuations in valuation distributions or order flow imbalance, if those
fluctuations vary quickly enough to neutralize over a duration of
order $L^2$ (where time is measured in volume offered). \footnote{Note that these $\sqrt{N}$ scales originate from the central limit (more specifically, Donsker's theorem)  and that is it a topic of further research to verify these exact scales empirically.}

Finally, note that the variance in (\ref{eq:highliq}) is not only
dependent on $\tau^2(x_E)$ in the way one might expect, but like
the location in (\ref{eq:highliq}), it is inverse proportional
to a liquidity-weighted convex combination of $f_A$ and $f_B$,
evaluated at the real equilibrium price. So the volatility of the
Brownian path (as well as the influence of excess liquidity) goes
down in ranges where orders are concentrated and goes up in ranges
where orders are sparse. Consequently, the Brownian path has long
occupation times in ranges where orders are dense. 

\begin{remark}
\cite{Toke15} derives a normal asymptotic distribution of the
clearing price in a similar setting, under the assumption of Poisson order
flow (\ie\ $N \sim \text{Pois}(\lambda T)$, where $\lambda T \to \infty$,
$N_A \sim \Bin(N,\alpha)$) and $F_A=F_B=F$. The Poisson order flow with
$\lambda T \to \infty$ represents a fixed randomization of the deterministic
$N\to\infty$ discussed here. However in the proof he firstly considers
(in our notation) fixed $N$ and $N_A = \alpha N$ and finds an asymptotic
normal distribution for $X$, with mean $F^{-1}(1-\alpha)$ and standard
deviation $\sqrt{\alpha(1-\alpha)}f(F^{-1}(1-\alpha))^{-1}$.
Setting $F_A=F_B=F$ and $D=0$ in our result, the solution to
(\ref{eq:real_eq_price}) is $x_E = F^{-1}(1-\alpha)$, $\tau^2(x_E) = \alpha(1-\alpha)$ and
$\sigma(x_E) = \sqrt{\alpha(1-\alpha)}f(x_E)^{-1}$.
\end{remark}
%%%%%%%%%%%%%%%%%%%%%%%%%%%%%%%%%%%%%%%%%%%%%%%%%%%%%%%%%%%%%%%%%%%%%%%%%%%%%%%%

\section{Supply-demand distributions, price and volume}
\label{sec:supdem}

In theorem~\ref{thm:xvdist} we derived the joint distribution of the
clearing price $X$ and the corresponding transacted volume $V$, given
supply and demand distributions $F_A,F_B$ and volumes $N_A,N_B$. In this subsection we explore the dependence of the distribution
of $(X,V)$ on $F_A$ and $F_B$. We shall fix $N_A$ and $N_B$ as equal
constants ($N_A=N_B=50$ in the examples below). It is also recalled
that the distribution for $(X,V)$ was derived in a setting with a discrete
price axis $\calX$ with ticksize $\delta>0$ (below, we take $\delta =0.01$);
normal distributions are discretized accordingly.

\subsection{Varying consensus between bid- and ask-side}

The supply and demand valuation distributions $F_A,F_B$ express a difference of
opinion concerning the valuation of the asset. Sell orders are typically
priced higher than buy orders, a fact expressed in a stochastic
way, through equation (\ref{eq:stochorder}). We first consider how shifts of
locations for $F_A,F_B$ influence the joint distribution of clearing
price and volume.

We consider three different choices of the supply and demand distributions,
denoted $F_A,F_B$, $\bar{F}_A,\bar{F}_B$ and $\tilde{F}_A,\tilde{F}_B$:
\begin{equation}
  \label{eq:consensusthree}
  F_A=\Phi_{10.1,0.1},\,F_B = \Phi_{9.9,0.1},\quad
  \bar{F}_A=\Phi_{10.05,0.1},\,\bar{F}_B=\Phi_{9.95,0.1},\quad
  \tilde{F}_A = \tilde{F}_B=\Phi_{10,0.1}.
\end{equation}
The first case represents a relatively large difference between the
locations of supply and demand distributions, while the second case
represents a small difference, and the third complete consensus. In
all three cases, the real equilibrium price is $x_E=10$, however, as can be seen from the left panel of figure~\ref{fig:fa_fb_shift},
\begin{equation}
  \label{eq:stochorderedFs}
  F_A(x_E)<\bar{F}_A(x_E)<\tilde{F}_A(x_E).
\end{equation}
\begin{figure}[htb]
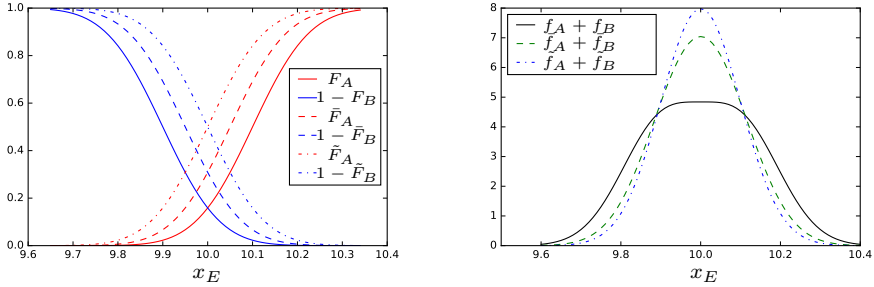

  \centering
    \parbox{13.5cm}{
      \begin{lpic}{fa_fb_shift(0.4)}
        \lbl[t]{122,70;{\tiny{$F_A$}}}	
        \lbl[t]{124,64;{\tiny{$1-F_B$}}}	
        \lbl[t]{122,58;{\tiny{$\bar{F}_A$}}}
        \lbl[t]{124,52;{\tiny{$1-\bar{F}_B$}}}		
        \lbl[t]{122,46;{\tiny{$\tilde{F}_A$}}}
        \lbl[t]{124,40;{\tiny{$1-\tilde{F}_B$}}}	 
        \lbl[t]{78,5;{\scriptsize{$x_E$}}}	     
      \end{lpic}
      \begin{lpic}{densitysums_case1(0.4)}
        \lbl[t]{45,88;{\tiny{$f_A+f_B$}}} 
        \lbl[t]{45,83;{\tiny{$\bar{f}_A+\bar{f}_B$}}} 
        \lbl[t]{45,77;{\tiny{$\tilde{f}_A+\tilde{f}_B$}}} 
        \lbl[t]{86,5;{\scriptsize{$x_E$}}}       
      \end{lpic}
      \caption{ \label{fig:fa_fb_shift} Left panel: distribution functions of supply/demand. Solid lines
      $F_A,F_B$; dashed lines $\bar{F}_A,\bar{F}_B$;
      dashed-dotted lines $\tilde{F}_A,\tilde{F}_B$. Right panel: sums of densities of supply/demand. Solid line $f_A+f_B$; dashed line $\bar{f}_A+\bar{f}_B$; dashed-dotted line $\tilde{f}_A+\tilde{f}_B$.}
      }
    \end{figure}

\begin{figure}
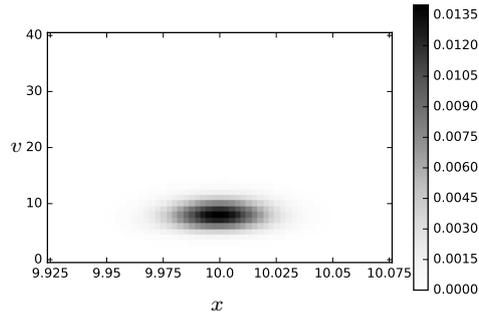
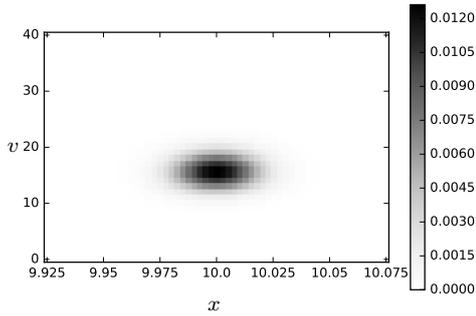
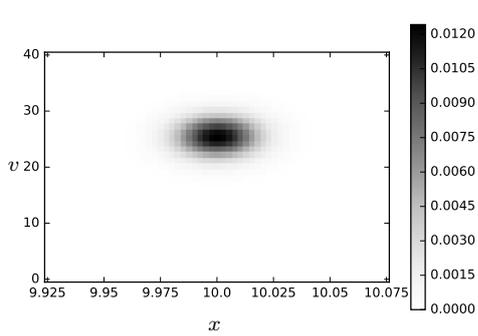

  \centering
    \parbox{13.5cm}{

    \centering
    \begin{subfigure}[bht]{0.4 \textwidth}
      \begin{lpic}{fxv_consensus1(0.48)}
        \lbl[t]{0,45;{\scriptsize{$v$}}}	
        \lbl[t]{55,1;{\scriptsize{$x$}}}	
      \end{lpic}
      \caption{Density for $(X,V)$ with valuation distributions $F_A=\Phi_{10.1,0.1},\,F_B = \Phi_{9.9,0.1}$.}	
    \end{subfigure}\\
    %  \hfill \\
    %  \hfill
    \begin{subfigure}[bht]{0.4 \textwidth}
      \begin{lpic}{fxv_consensus2(0.48)}
        \lbl[t]{0,45;{\scriptsize{$v$}}}	
        \lbl[t]{55,1;{\scriptsize{$x$}}}			
      \end{lpic}
	    \caption{Density for $(X,V)$ with valuation distributions $\bar{F}_A=\Phi_{10.05,0.1},\bar{F}_B=\Phi_{9.95,0.1}.$}     
    \end{subfigure}
    ~
    \begin{subfigure}[bht]{0.4 \textwidth}
      \begin{lpic}{fxv_consensus3(0.48)}
        \lbl[t]{0,45;{\scriptsize{$v$}}}	
        \lbl[t]{55,1;{\scriptsize{$x$}}}			
      \end{lpic}
      \caption{Density for $(X,V)$ with valuation distributions $\tilde{F}_A = \tilde{F}_B=\Phi_{10,0.1}.$} 
    \end{subfigure}
    \caption{\label{fig:fxv_shift} The influence of consensus
    between bid- and ask-side of the market on the distribution
    of $(X,V)$. Note: as the locations
    of supply and demand distributions diverge, transacted volume
    drops, while price uncertainty increases.}
    }
\end{figure}
Figure~\ref{fig:fxv_shift} shows the distributions of
price-volume in these three cases and suggests the following,
intuitively reasonable mechanism: as the locations of supply and
demand distributions diverge, marginally the transacted volume drops,
while the width of the price marginal increases. Note that the
location on the price-axis
does not change, as all three $(X,V)$-distributions are centred
around $x_E=10$. Referring to theorem~\ref{thm:highliq}, the
result reflects the ordering expressed by (\ref{eq:stochorderedFs}):
in the high-liquidity limit, $X$ lies close to $x_E$ and
$V=N_A \FA(X)$ (respectively, $\bar{V}=N_A \bar{\mathbb{F}}_A(X)$,
$\tilde{V}=N_A \tilde{\mathbb{F}}_A(X)$) lies close to $N_A F_A(x_E)$
(respectively, $N_A \bar{F}_A(x_E)$, $N_A \tilde{F}_A(x_E)$).
Similar arguments regarding the ordering of densities (see also the right panel of figure~\ref{fig:fa_fb_shift}),
\[
  f_A(x_E)+f_B(x_E)
    < \bar{f}_A(x_E)+\bar{f}_B(x_E)
    < \tilde{f}_A(x_E)+\tilde{f}_B(x_E),
\]
 provide an asymptotic explanation for the observed increase
in price uncertainty (\cf\ the denominator of the variance in
(\ref{eq:highliq}); the numerator is bounded and plays no role
here). To re-phrase and summarize: when consensus
between bid- and ask-sides increases, transacted volume increases
and price uncertainty decreases.

\subsection{Increased uncertainty among market participants}

Here we investigate the influence of valuation uncertainty among market participants
on the distribution of clearing price and transacted volume: we
consider three different choices of the supply and demand distributions,
denoted $F_A,F_B$, $\bar{F}_A,\bar{F}_B$, $\tilde{F}_A,\tilde{F}_B$,
\begin{equation}
  \label{eq:uncertaintythree}
  F_A=\Phi_{10.1,0.1},\,F_B=\Phi_{9.9,0.1},\quad
  \bar{F}_A=\Phi_{10.1,0.2},\,\bar{F}_B=\Phi_{9.9,0.2},\quad
  \tilde{F}_A=\Phi_{10.1,0.3},\,\tilde{F}_B=\Phi_{9.9,0.3}.
\end{equation}
The locations of supply and demand distributions are maintained,
while their variances are increased, reflecting growing uncertainty in valuation among individual market participants. Again, in all three cases, the real equilibrium price is
$x_E=10$ and (\ref{eq:stochorderedFs}) continues to hold (see the left panel of figure \ref{fig:fa_fb_scale}).
\begin{figure}[htb]
  \centering
    \parbox{13.5cm}{
      \begin{lpic}{fa_fb_scale(0.4)}
        \lbl[t]{36,61;{\tiny{$F_A$}}} 
        \lbl[t]{40,55;{\tiny{$1-F_B$}}} 
        \lbl[t]{36,50;{\tiny{$\bar{F}_A$}}}
        \lbl[t]{40,44;{\tiny{$1-\bar{F}_B$}}}   
        \lbl[t]{36,38;{\tiny{$\tilde{F}_A$}}}
        \lbl[t]{40,32;{\tiny{$1-\tilde{F}_B$}}}  
        \lbl[t]{78,5;{\scriptsize{$x_E$}}}       
      \end{lpic}
      \begin{lpic}{densitysums_case2_v2(0.4)}
        \lbl[t]{123,89;{\tiny{$f_A+f_B$}}} 
        \lbl[t]{123,83;{\tiny{$\bar{f}_A+\bar{f}_B$}}} 
        \lbl[t]{123,77;{\tiny{$\tilde{f}_A+\tilde{f}_B$}}} 
        \lbl[t]{123,70;{\tiny{$\munderbar{f}_A+\munderbar{f}_B$}}} 
        \lbl[t]{78,5;{\scriptsize{$x_E$}}}       
      \end{lpic}
      \caption{ \label{fig:fa_fb_scale} Left panel: distribution functions of supply/demand. Solid lines
      $F_A,F_B$; dashed lines $\bar{F}_A,\bar{F}_B$;
      dashed-dotted lines $\tilde{F}_A,\tilde{F}_B$. Right panel: sums of densities of supply/demand. Solid line $f_A+f_B$; dashed line $\bar{f}_A+\bar{f}_B$; dashed-dotted line $\tilde{f}_A+\tilde{f}_B$; dotted line $\munderbar{f}_A+\munderbar{f}_B$.}
      }
    \end{figure}

\begin{figure}[htb]
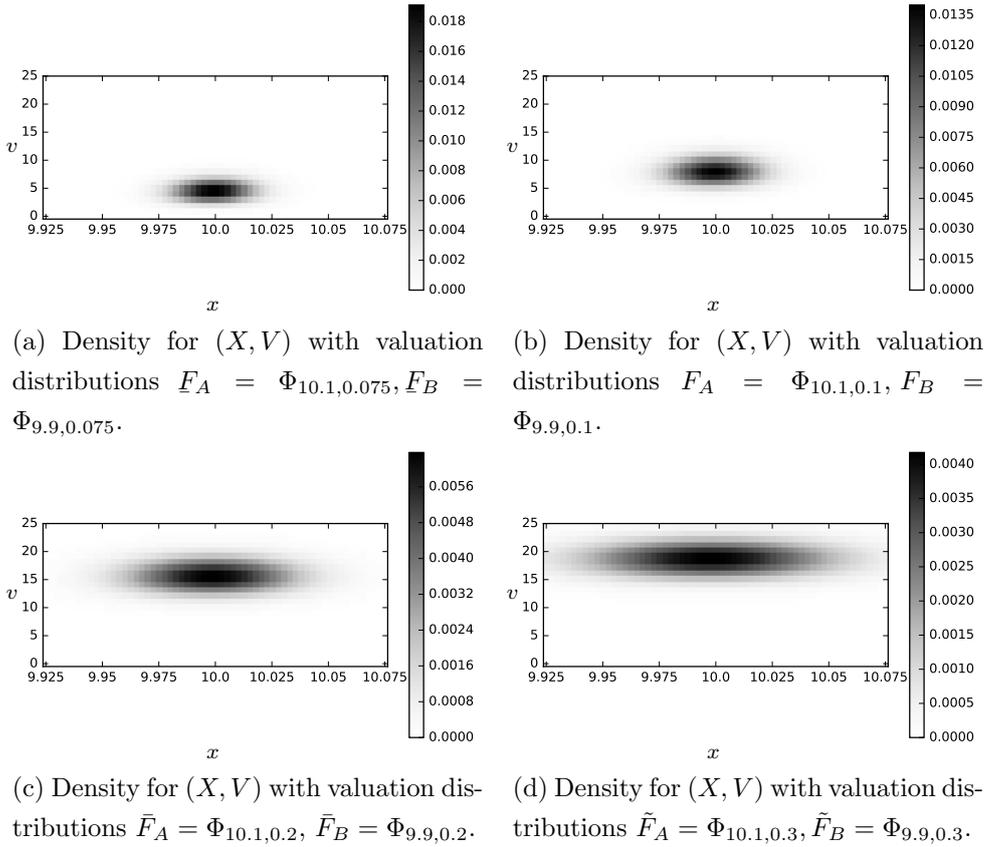

  \centering
    \parbox{13.5cm}{
    \begin{subfigure}[bht]{0.4 \textwidth}
      \begin{lpic}{fxv_uncertainty4_v2(0.48)}
        \lbl[t]{0,45;{\scriptsize{$v$}}}  
        \lbl[t]{55,1;{\scriptsize{$x$}}}  
      \end{lpic}
      \caption{Density for $(X,V)$ with valuation distributions $\munderbar{F}_A=\Phi_{10.1,0.075},\munderbar{F}_B=\Phi_{9.9,0.075}$.} 
    \end{subfigure}
    ~
    \begin{subfigure}[bht]{0.4 \textwidth}
      \begin{lpic}{fxv_uncertainty_v2(0.48)}
        \lbl[t]{0,45;{\scriptsize{$v$}}}  
        \lbl[t]{55,1;{\scriptsize{$x$}}}  
      \end{lpic}
      \caption{Density for $(X,V)$ with valuation distributions $F_A=\Phi_{10.1,0.1},\,F_B=\Phi_{9.9,0.1}$.~~~~~~~~~} 
    \end{subfigure}\\
    %  \hfill \\
    %  \hfill
    \begin{subfigure}[bht]{0.4 \textwidth}
      \begin{lpic}{fxv_uncertainty2_v2(0.48)}
        \lbl[t]{0,45;{\scriptsize{$v$}}}  
        \lbl[t]{55,1;{\scriptsize{$x$}}}      
      \end{lpic}
      \caption{Density for $(X,V)$ with valuation distributions $\bar{F}_A=\Phi_{10.1,0.2},\,\bar{F}_B=\Phi_{9.9,0.2}$.}     
    \end{subfigure}
    ~
    \begin{subfigure}[bht]{0.4 \textwidth}
      \begin{lpic}{fxv_uncertainty3_v2(0.48)}
        \lbl[t]{0,45;{\scriptsize{$v$}}}  
        \lbl[t]{55,1;{\scriptsize{$x$}}}      
      \end{lpic}
      \caption{Density for $(X,V)$ with valuation distributions $\tilde{F}_A=\Phi_{10.1,0.3},\tilde{F}_B=\Phi_{9.9,0.3}$.} 
    \end{subfigure}
    \caption{\label{fig:fxv_scale} The influence of valuation uncertainty among market participants on the distribution of $(X,V)$. }
    }
\end{figure}
Panels (b), (c) and (d) of figure~\ref{fig:fxv_scale} show the distributions of
price-volume in these three cases and suggests the following,
reasonable-sounding (but incomplete, see below) rule: as the
variance of the valuation distributions increases,
marginally both transacted volume and price uncertainty increase.
Note that the location on the price-axis does not change, as all
three $(X,V)$-distributions have marginals centred around $x_E=10$.
The asymptotic argument for the observed ordering
$V < \bar{V} < \tilde{V}$ continues to
hold. Note in the right panel of figure \ref{fig:fa_fb_scale}, however, that with locations and variances as chosen,
\[
  f_A(x_E)+f_B(x_E)
    > \bar{f}_A(x_E)+\bar{f}_B(x_E)
    > \tilde{f}_A(x_E)+\tilde{f}_B(x_E),
\]
so that asymptotic variance of the clearing price \emph{increases}
when valuations become more widely spread
(referring again to the variance in (\ref{eq:highliq})).

But this explains only half of the mechanism that the model ascribes
to the relation between valuation uncertainty and auction price variance.
To appreciate the other half, consider a fourth different pair $\munderbar{F}_A,\munderbar{F}_B$  of supply and demand distributions that reflects less valuation uncertainty among market participants, defined by
$$\munderbar{F}_A = \Phi_{10.1,0.075},~ \munderbar{F}_B = \Phi_{9.9,0.075}.$$
As can be seen from the right panel of figure \ref{fig:fa_fb_scale}, this choice of valuation distributions satisfies
$$\munderbar{f}_A(x_E) + \munderbar{f}_B(x_E) < f_A(x_E)+f_B(x_E),$$
which implies that the asymptotic variance of the clearing price also increases when we lower the variance of the valuation distributions. This is also confirmed by panel (a) of figure \ref{fig:fxv_scale}, where the distribution of price-volume for $\munderbar{F}_A,\munderbar{F}_B$ is shown.
 To explain this observed inversion, consider $F_A$ and $F_B$ that
are two normal distributions of equal variance $\sigma^2>0$,
located at $\mu_1,\mu_2\in\RR$. Reasoning again asymptotically,
the denominator of the expression for the variance in
(\ref{eq:highliq}) equals,
\begin{equation}
  \label{eq:inversevol}
  f_A(x_E)+f_B(x_E)=
  \sqrt{\frac{2}{\pi}}\frac1{\sigma}
    \exp\Bigl(-\frac12\frac{(\mu_1-\mu_2)^2}{\sigma^2}\Bigr).
\end{equation}
As a function of $\sigma$, (\ref{eq:inversevol}) has a \emph{maximum}
at $\sigma=\ft12|\mu_1-\mu_2|$ (see figure~\ref{fig:densitysums_vs_sigma} for an example), which means that asymptotic variance
of the clearing price is \emph{minimal} at said level of valuation uncertainty
$\sigma$. When $\sigma$ rises above
$\ft12|\mu_1-\mu_2|$, as in figure~\ref{fig:fxv_scale}, panels (c) and (d), auction
price variance increases; perhaps somewhat surprisingly, when $\sigma$  decreases below $\ft12|\mu_1-\mu_2|$, as in figure~\ref{fig:fxv_scale} panel (a), auction
price variance also increases.
The heuristic reason for this inversion is as follows: when
consensus between the bid- and ask-side of market is very low (large $|\mu_1-\mu_2|$) and valuation uncertainty among market participants is minimal (small $\sigma$), orders around $x_E$
are very scarce, so that clearing prices are based on \emph{small
numbers of matchable orders}, therefore displaying \emph{high
variance}; as the uncertainty in order prices on both sides increases,
more orders appear around $x_E$, lowering the variance of the
clearing price. The added valuation uncertainty `unlocks' an
otherwise \emph{illiquid market}, in which buyers and
sellers rarely cross. So in a market with illiquidity-driven
price movements, raised valuation uncertainty aids accurate price
discovery.

Combination with the previous subsection invites the following,
intuitively reasonable conclusion: \emph{observation of high
levels of price variance can be driven by illiquidity or
by valuation uncertainty among market participants; observation of the price and its
fluctuations alone does not distinguish between those cases.
To differentiate one must involve transacted volume, which is
moderate when price variance is minimal, low
in illiquid markets and high in markets with valuation uncertainty-driven price variance}.
\begin{figure}[h]
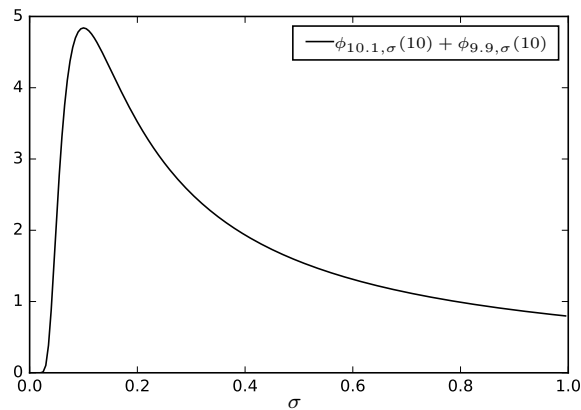

\centering
  \parbox{13.5cm}{
    \centering
    \begin{lpic}{densitysums_vs_sigma(0.6)}
      \lbl[t]{110,87;{\tiny{$\phi_{10.1,\sigma}(10)+\phi_{9.9,\sigma}(10)$ }}}
        \lbl[t]{78,7;{\scriptsize {$\sigma$} }}    
    \end{lpic}
    \caption{\label{fig:densitysums_vs_sigma} $f_A(x_E)+f_B(x_E)$ as a function of $\sigma$, for $f_A = \phi_{\mu_1,\sigma}, f_B =\phi_{\mu_2,\sigma}$, where $\mu_1=10.1, \mu_2=9.9$ and $x_E=10$. Note this function attains its maximum at $\sigma = \frac12|\mu_1-\mu_2| = 0.1$, implying that the asymptotic clearing price variance of equation \ref{eq:highliq} attains its minimum at this $\sigma$.}
    }
    \end{figure}
%%%%%%%%%%%%%%%%%%%%%%%%%%%%%%%%%%%%%%%%%%%%%%%%%%%%%%%%%%%%%%%%%%%%%%%%%%%%%%%%

\section{Impact of market orders}
\label{sec:impact}

In definition~\ref{def:refprice_el} the clearing price in the presence
of excess liquidity $\Delta$ is defined and its distribution is
provided in proposition~\ref{prop:eqprice_el_dist}. Modelling
the arrival of market orders as excess liquidity, this subsection
compares clearing prices with and without market orders.
Differences between clearing price distributions form the
model's perspective on the \emph{price impact of market orders},
a subject that has received quite some attention in the
literature (see \eg\ \cite{Hasbrouck91,Lilloetal03, Smithetal03,Donierbonart15,Donieretal15,Benzaquenbouchaud18} and references therein).

Consider again the case that $F_A=F_B=\Phi_{\mu,\sigma}$,
for $\mu=10,~\sigma=0.1$ and $(N_A,N_B) \sim \text{Pois}(50)^2$.
Departing from the case that this market
is in equilibrium, next suppose that a market order of size
$|\omega|$ arrives: as in eq.~(\ref{eq:refprice_el}), we add an
\emph{excess liquidity} term to model this, in the form of
constant functions $\Delta(x) = \omega$, where $\omega>0$ corresponds
to a buy order and $\omega<0$ represents a sell order.
\begin{figure}[htb]
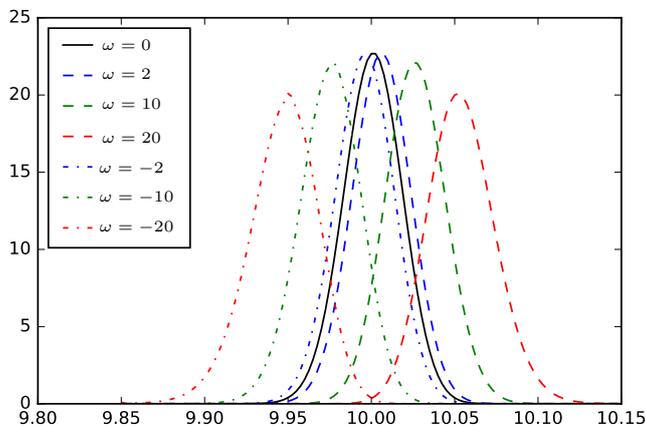

  \centering
    \parbox{13.5cm}{
    \centering
    \begin{lpic}{el_market_order_poiss_v2(0.65)}
      \lbl[t]{37,87;{\tiny{$\omega = 0$}}}
      \lbl[t]{37,81;{\tiny{$\omega = 2$}}}
      \lbl[t]{38,75;{\tiny{$\omega = 10$}}}
      \lbl[t]{38,68;{\tiny{$\omega = 20$}}}
      \lbl[t]{38,62;{\tiny{$\omega = -2$}}}
      \lbl[t]{39,56;{\tiny{$\omega = -10$}}}
      \lbl[t]{39,50;{\tiny{$\omega = -20$}}}    
    \end{lpic}
    \caption{\label{fig:marketorder_pois} Clearing price densities, when
    market orders of sizes $|\omega|$ are placed (negative $\omega$
    for sell orders, positive for buy orders). The supply and demand
    distributions are equal and normal, $F_A=F_B=\Phi_{\mu,\sigma}$, for
    $\mu = 10,~\sigma =0.1$, while $(N_A,N_B)\sim\text{Pois}(50)^2$. Note that orders of size $|\omega|=2$ do not
significantly influence the price distribution, but orders of sizes
$|\omega|= 10$ or $20$ shift the clearing price distribution noticeably.}
    }
\end{figure}
In figure~\ref{fig:marketorder_pois} the resulting clearing price
distributions are 
plotted for various $\omega$.
\begin{figure}
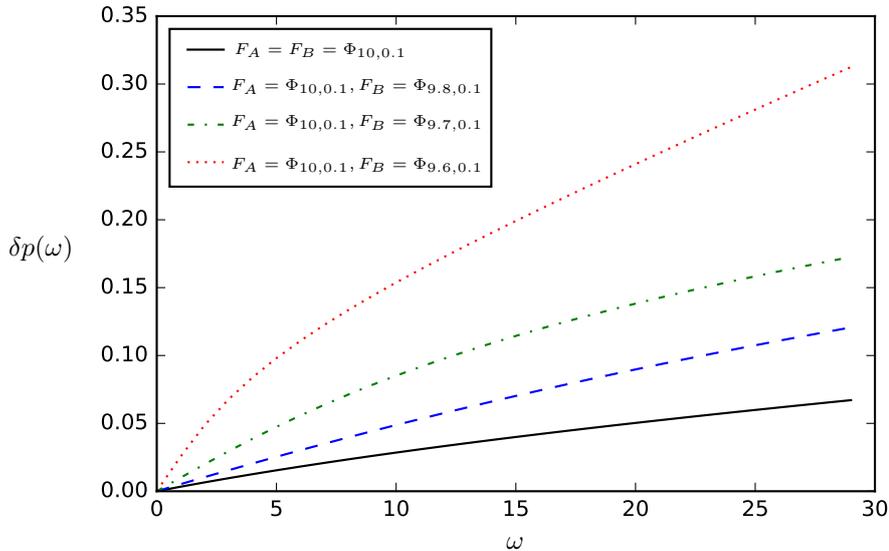

  \centering
    \parbox{13.5cm}{
    \centering
    \begin{lpic}{price_impact_beta(0.8)}
      \lbl[t]{46,87;{\tiny{$F_A=F_B=\Phi_{10,0.1}$}}}
      \lbl[t]{52,81;{\tiny{$F_A = \Phi_{10,0.1},F_B=\Phi_{9.8,0.1}$}}}
      \lbl[t]{52,75;{\tiny{$F_A = \Phi_{10,0.1},F_B=\Phi_{9.7,0.1}$}}}
      \lbl[t]{52,68;{\tiny{$F_A = \Phi_{10,0.1},F_B=\Phi_{9.6,0.1}$}}}
      \lbl[t]{0,55;{\small{$\delta p(\omega)$}}}
      \lbl[t]{78,5;{\small{$\omega$}}}
    \end{lpic}
    \caption{\label{fig:priceimpact} Price impact $\delta p$ as a function
    of the size $\omega$ of the market order, for various supply and
    demand distributions $F_A,F_B$. In all cases $N=100$, $N_A=\alpha N$,
    $\alpha\sim\text{Beta}(2,2)$.}
    }
\end{figure}
The common definition of the \emph{price impact function} $\delta p(\omega)$ is
the size of the shift in market price when a market order of size
$\omega$ arrives. Empirical studies (see \eg\ \cite{Donierbonart15,Hasbrouck91,Lilloetal03}, among many others) have shown that, in the situation of a continuous
double auction, the price impact function is concave, and certain
models confirm this concavity (see \eg\ the seminal paper by \cite{Smithetal03}, or more recent work in this area \cite{ Benzaquenbouchaud18,Donieretal15}). To consider the matter in our model, we
define the price impact function $\delta p(\omega)$ as the shift in
\emph{expectation} of $X$ when a buy market order of size
$\omega>0$ arrives.
Figure~\ref{fig:priceimpact} shows price impact functions for various
supply and demand distributions that display the expected concavity.
Furthermore, the picture shows that price impact becomes less
concave as supply and demand distributions are shifted together,
with the case $F_A=F_B$ almost linear. This difference is explained
by the number of orders that can be expected around $\mathbb{E}X$.
In the case $F_A=F_B$, $\mathbb{E}X$ lies around $x_E=10$ and all orders lie
around $10$. In cases where the locations of $F_A$ and $F_B$ differ,
$\mathbb{E} X$ lies between them, while buy orders concentrate
around a lower price and sell orders around a higher one. In that
situation fewer orders lie around $\mathbb{E}X$ and consequently the
clearing price is impacted more significantly in such regions;
by contrast, in regions where orders are more concentrated, the
clearing price is less easily moved. Comparing with
\cite{Smithetal03}, their model produces an almost linear price impact function for a situation in which there is a large accumulation of orders near the market price and a very concave price impact function with lower levels of accumulation near the market price.

%%%%%%%%%%%%%%%%%%%%%%%%%%%%%%%%%%%%%%%%%%%%%%%%%%%%%%%%%%%%%%%%%%%%%%%%%%%%%%%%

\section{Prediction of the closing price distribution}
\label{sec:application}

For a quantitative model, a convincing statistical demonstration of
applicability is ultimately the only possible proof of relevance.
Below we perform this statistical exploration:
we consider the statistical quality of the model's clearing
price distributions in daily closing auctions for five (randomly
selected) Eurostoxx 50 index constituents with the Kolmogorov-Smirnov
goodness-of-fit test, and find that they explain the
randomness in observed closing prices well. More specifically, we
use a day's transactions to estimate clearing price distributions
for daily closing auctions of five shares over the course of the
trading year 2017. We assume that we have observed the market until 5 pm and then want to predict the closing price distribution.\footnote{The choice of the prediction time of 5 pm is not completely arbitrary. We have found empirically that around 90\% of the closing prices falls within a range of 30bps of the last mid-price, and that the closing price is generally very close to the last mid-price. Hence, there is not much to predict when we wait until the closing auction starts, as the last mid-price is then more informative than our prediction. Of course we could start prediction already before 5 pm, but then the quality of the estimators will get worse, as less transaction data is observed. }To assess performance, we keep track of the
quantiles of realized closing prices according to the estimated
clearing price distribution: if the estimates are accurate (and
approximately independent), these quantiles form an approximate
\iid\ sample from the uniform distribution on $[0,1]$. The match
is assessed graphically, through QQ-plots, and tested with the
Kolmogorov-Smirnov statistic.
As a simple benchmark, the results are compared with results from
a log-normal model.

\subsection{Estimation of the closing price distribution}

To obtain the daily estimator for the clearing price distribution, we first need estimators for the supply- and demand-distributions $F_A$ and $F_B$. As we want to predict the closing price distribution $F_X$ before the start of the closing auction, it is not an option to use quote data from the closing auction itself. Instead, intra-day transaction data is used: throughout the trading day, all
transactions are recorded in a book that aggregates total
volume traded for any price tick in the daily price range. In
fact, two such books are kept, distinguished by the side of
the market that initiated the trade. Half an hour before
market close these books are normalized and converted 
into histogram-like estimators for the densities $f_A$ and
$f_B$. Expressed cumulatively, this leads to `empirical
distributions functions' $\hat{\mathbb{F}}_A(\cdot)$ and
$\hat{\mathbb{F}}_B(\cdot)$ that serve as estimators for
$F_A$ and $F_B$. Essentially we use a volume-weighted version
of the day's \emph{transacted} orders to estimate market participants'
valuations. This leads to reasonable estimators, based on the idea that the intra-day valuations of market participants
will be reflected in their valuations in the closing auction.

For any choice of $N_A,N_B$, these daily estimators can be used to
estimate the distribution for $X|N_A,N_B$, that day's clearing
price given order flow $N_A,N_B$. Because
$\hat{\mathbb{F}}_A(\cdot)$ and $\hat{\mathbb{F}}_B(\cdot)$ are
supported on the range of prices visited that day, the clearing
price distribution is supported on that range too. This causes a disadvantage of the proposed equilibrium model:
regardless of the order flow, the model does not predict
anything outside the daily price range and the estimator should
be viewed as a `windowed' or conditioned device, relevant only
conditional on auction prices that fall inside the daily price
range.
To model order flow, we convolute with an order flow marginal
in which $N$ is fixed and $N_A=\alpha N$, with $\alpha$
distributed according to a Beta-distribution. Because we have
no reason to assume asymmetry, only the scale $\beta>0$
in $\alpha\sim\text{Beta}(\beta,\beta)$ varies. To not exclude
the possibility of fairly extreme, one-sided order flow
(where $N_A\gg N_B$ or vice versa with high probability),
we keep $\beta<1$ (this is empirically verified, see remark \ref{rem:betaN} below).
\begin{definition} \label{def:estcpdist} The daily
\emph{estimated clearing price distribution} is the
distribution $\hat{\mathbb{F}}_X$ that results from
theorem~\ref{thm:cpdist} with empirical supply and demand
$F_A=\hat{\mathbb{F}}_A$, $F_B=\hat{\mathbb{F}}_B$, convoluted
with the order flow distribution.
\end{definition}
In all examples below, we choose $N=100$ and $\beta=0.75$ and
note that these choices appear to work well for the five
Eurostoxx 50 index constituents considered below.
\begin{remark} \label{rem:betaN}
It is important to note that the model is robust with respect to these choices. The choice of $N$ is not very important: at first sight $N$ has a very clear interpretation as the number of orders in the auction, however it should be noted that the order size of orders is not modelled, hence changing $N$ could also be interpreted as changing the order size. As long as $N$ is taken sufficiently large to allow for enough diversification in the orders (already for $N>50$) the choice of $N$ does not really affect the result,  as the large liquidity limit starts to do its work. \footnote{There is also an easy way to estimate $N$: around 28\% of the daily total transacted volume is transacted in the closing auction (in 2017, nowadays it is more), so one could take the total transacted volume until 5 pm and turn this into an estimator for $N$ using this ratio. However, figure \ref{fig:roleN} shows that this analysis is not worthwile, as a different choice of $N$ does not impact the results.} These claims are supported by figure \ref{fig:roleN}. The choice of $\beta$ has more influence on the results, as it determines how many mass is shifted into the tails of the distribution. But still, the model is robust with respect to small variations in $\beta$, which is expressed by figure \ref{fig:rolebeta}. Estimation of the parameter $\beta$ is also interesting: one can take the $n$ previous closing auctions of the stock and compute the ratio $\frac{N_A}{N_A+N_B}$, \ie\ the total volume of all sell orders in the closing auction, divided by the total volume of all orders. \footnote{In fact, the order book contains a lot of irrelevant volume far from the eventual closing price. This volume does not contribute to the determination of the closing price and should not be counted in the estimation. Instead, we only counted orders within ten levels of the closing price.} Then one obtains a sample of $n$ order flow imbalances, which can be treated as an \iid\ sample to find the parameter of the Beta-distribution, by computing moment estimators. We did this analysis for three of the five Eurostoxx 50 index constituents considered below. \footnote{One needs full order book data to do the estimation, which is provided for the stocks traded on Euronext, but not for the German stocks Bayer AG and Deutsche Telekom AG.} For Airbus SE we found $\hat{\beta} = 0.8057$, for Engie SA we found $\hat{\beta} = 0.8219$ and for Anheuser Busch Inbev NV we found $\hat{\beta}=0.6988$, empirically justfying the observation that $\beta$ should be picked around 0.75.
\end{remark} 
\begin{figure}[htb]
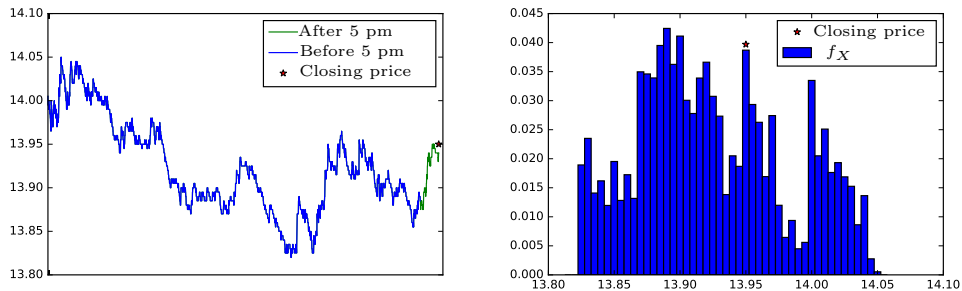

  \centering
  \parbox{13.5cm}{
    \begin{subfigure}{0.4\textwidth}
      \begin{lpic}{example_realprices(0.44)}
        \lbl[t]{108,88;{\tiny{After 5 pm}}} 
        \lbl[t]{110,82;{\tiny{Before 5 pm}}}  
        \lbl[t]{110,76;{\tiny{\,\,Closing price}}}      
      \end{lpic}
    \end{subfigure}
    ~
    \begin{subfigure}[bht]{0.4\textwidth}
      \begin{lpic}{example_realprices_df(0.44)}
        \lbl[t]{114,88;{\tiny{\,\,Closing price}}}  
        \lbl[t]{106,82;{\tiny{$f_X$}}}  
      \end{lpic}
    \end{subfigure}
    \caption{\label{fig:realprices} Estimation of the closing price
    distribution, based on daily transaction data. Left panel: ING stock
    price during a trading day in 2017. Right panel: the estimated
    distribution of the closing price $X$, where the distributions of
    supply and demand are taken to be the empirical estimators of
    $F_A$ and $F_B$ (based on the data until 5 pm) and where
    $N=100$, $N_A = \alpha N$, $\alpha\sim \text{Beta}(0.75,0.75)$.}
  }
\end{figure}
As an example, consider figure~\ref{fig:realprices},
an (arbitrarily selected) day's trading in ING stocks and the
estimate of the closing price distribution $\hat{\mathbb{F}}_X$ (based on $F_A,~F_B$ that are estimated from daily transaction data until 5 pm). Note the inhomogeneity
of the estimated density. The following statistical analysis shows that this detailed shape with peak and troughs is informative for the realized closing price, meaning that the closing price is more likely to be realized on prices where the estimated density is higher, which is nicely illustrated in the above example by the closing price that realizes on one of the peaks in the estimated density.

\subsection{Kolmogorov-Smirnov goodness-of-fit test}

To assess performance, we predict closing price distributions
for the circa 250~trading days in 2017. Because the model only
concerns the daily trading range, we do not include those trading
days on which the stock's closing price lay \emph{outside} the
daily trading range. Typically, trading on such days is highly
momentum-driven and is not well-represented by equilibrium models, at least, on daily or shorter time-scales. After removing the days where the price closed outside the daily trading range, this leads to samples of 200-230 trading days for the selected
five stocks. As the valuation distributions $F_A,F_B$ differ from day
to day, there is no straightforward way to assess the accuracy
of the sample of estimated clearing price distributions
$\hat{\mathbb{F}}_X$. 
For that, we need a standard, distribution-free argument based
on the observation that if $X\sim F_X$, then $F_X^{-1}(X)\sim U[0,1]$:
if $\hat{\mathbb{F}}_X$ approximates $F_X$ well on any trading
day, $\hat{\mathbb{F}}_X^{-1}(X)$ has a distribution approximating
$U[0,1]$. In our statistical
experiment we have a sequence of predictions $\hat{\mathbb{F}}_i$
for closing prices $X_i$ (assumed independent) with true marginal
closing price distributions $F_i$ (which are possibly
very different as the day $i$ varies). If
the estimators $\hat{\mathbb{F}}_i$ approximate the $F_i$ well,
the resulting sequence $\hat{\xi}_i=\hat{\mathbb{F}}_i^{-1}(X_i)$ is
distributed approximately as an \iid\ sample from the uniform
distribution on $[0,1]$. Below, this degree of approximation is
assessed graphically through QQ-plots and tested with statistical
significance using the Kolmogorov-Smirnov(KS) statistic.
This statistical assessment is not just a technically
convenient choice, what is assessed in this way is highly relevant
to daily market practice: good QQ-plots and
KS-statistics indicate that clearing price distribution estimators provide
an accurate picture of the relation between quoted price and
probability of execution in the auction (conditional on a closing price inside the price range seen during the day). For example, from a trader's perspective, the estimated quantiles could give rise to a trading strategy that goes long/short when the market price lies in the low/high quantiles half an hour before the market closes. From an investor's point of view, the $p$th percentile of the clearing price distribution answers the question at which price to quote in the auction to be for $p$\% sure that the order gets transacted.

To have a simple benchmark for comparison we also consider an
alternative: we include a benchmark model that assumes that the daily
log-return is normally distributed, with mean and variance estimated
by (volume-weighted) average and variance of log-prices of transactions
during the day. The resulting estimated closing price distribution
$\tilde{\mathbb{F}}_{i}$ is truncated to that day's trading
range. This leads to samples of $200-230$ quantiles $\tilde{\xi}_{i}=\tilde{\mathbb{F}}_{i}^{-1}(X_i)$, subject to the same requirement of similarity to
an \iid\ sample from the $U[0,1]$-distribution.
The two samples $\hat{\xi}_i$ (resulting from the market clearing model) and $\tilde{\xi}_{i}$ (resulting from the log-normal model)
are assessed for uniformity by QQ-plots in figure \ref{fig:qq}.

Table~\ref{table:ks} reports the associated KS-statistics and $p$-values. (Note that the KS-test does not fall within
the standard Neyman-Pearson framework of statistical testing, basically
because one seeks to \emph{confirm} the null-hypothesis. This changes
the usual interpretation of $p$-values: if a model
has a low $p$-value in this context, the hypothesis that it is correct
is rejected based on the data with high statistical
significance. By contrast, a model with a high $p$-value requires
a high degree of relaxation of significance criteria before the
correctness hypothesis is rejected based on the data.)
\begin{table}[thb]
\caption{ Kolmogorov-Smirnov statistics and
    corresponding $p$-values for the samples of quantiles
    $\hat{\xi}_i$ (resulting from the clearing model) and $\tilde{\xi}_{i}$ (resulting from the log-normal model)
    for five (randomly selected) constituents of the Eurostoxx 50 index.}
  
  \parbox{0.85\textwidth}{
  \begin{tabular}{|l|l|l|l|}
    \hline
    Stock & Model & KS-statistic    & $p$-value     \\ \hline
    1. Engie SA & Market clearing    & 0.0392 & 0.905 \\ \hline
      & Log-normal & 0.147 & 2.04*$10^{-4}$ \\ \hline
    2. Airbus SE &Market clearing  & 0.0320 & 0.988 \\ \hline
    & Log-normal & 0.164 & 3.76*$10^{-5}$ \\ \hline
    3. Bayer AG & Market clearing   & 0.0326 & 0.983 \\ \hline
      & Log-normal & 0.139 & 6.96*$10^{-4}$ \\ \hline
    4. Anheuser Busch Inbev NV & Market clearing&0.104&0.0198\\ \hline 
     & Log-normal & 0.175 & 4.61*$10^{-6}$ \\ \hline
    5. Deutsche Telekom AG & Market clearing& 0.0304 & 0.989  \\ \hline
    & Log-normal & 0.0837 & 0.0969 \\ \hline
    \end{tabular}
    }
    \label{table:ks}
\end{table}
The model of log-normal daily returns proves
wholly inadequate as an explanation of the randomness observed in
actual closing prices: only in the example of Deutsche Telekom is
it possible to argue that (truncated) log-normal distributions for the daily returns form
a prediction that is informative about closing prices at the
distributional, predictive level. Furthermore, figure~\ref{fig:qq} shows that the log-normal model
underestimates the tail of the closing price distributions in all five examples.
By contrast, the QQ-plots for the market clearing model show
a very decent match for uniformity, indicating that the model
is a good representation of the randomness in observed
closing prices. Estimated clearing price distributions
provide an accurate picture of the relation between quoted price and probability of execution (conditional on
an auction price that falls within the daily trading range).
This is confirmed by associated KS-statistics and their
$p$-values in table~\ref{table:ks}: four out of five
samples show exceptionally straight lines in their QQ-plots,
confirmed by the exceptionally high $p$-values in Table
\ref{table:ks}. The exception (recall, these five stocks
have been selected randomly from the Eurostoxx 50 index) is
the Anheuser Busch Inbev NV stock, with a KS-statistic that
indicates evidence ($p=0.0198$) to reject the null-hypothesis
and visual inspection through the QQ-plot reveals
underestimation of the up-side tail.
\begin{figure}[t!]
\captionsetup[subfigure]{justification=justified,singlelinecheck=false}
  \centering
  \parbox{0.85\textwidth}{
    \begin{subfigure}[t]{0.4\textwidth}
      \begin{lpic}{qq_nolabel_388276(0.30)}  
      \end{lpic}
      \subcaption*{\hspace{0.1cm}1a. Engie SA, market clearing}
    \end{subfigure}
    ~
    \begin{subfigure}[t]{0.4\textwidth}
      \begin{lpic}{qq3_388276(0.30)}  
      \end{lpic}
      \subcaption*{\hspace{0.1cm}1b. Engie SA, log-normal}
    \end{subfigure}
    \\
    \begin{subfigure}[t]{0.4\textwidth}
      \begin{lpic}{qq_nolabel_85691(0.30)}  
      \end{lpic}
      \subcaption*{\hspace{0.1cm}2a. Airbus SE, market clearing}
    \end{subfigure}
    ~
    \begin{subfigure}[t]{0.4\textwidth}
      \begin{lpic}{qq3_85691(0.30)}  
      \end{lpic}
      \subcaption*{\hspace{0.1cm}2b. Airbus SE, log-normal}
    \end{subfigure}
    \\
    \begin{subfigure}[t]{0.4 \textwidth}
      \begin{lpic}{qq_nolabel_62(0.30)}
      \end{lpic}
      \subcaption*{\hspace{0.1cm}3a. Bayer AG, market clearing}
    \end{subfigure}
    ~
    \begin{subfigure}[t]{0.4 \textwidth}
      \begin{lpic}{qq3_62(0.30)}
      \end{lpic}
      \subcaption*{\hspace{0.1cm}3b. Bayer AG, log-normal}
    \end{subfigure}
    \\
    \begin{subfigure}[t]{0.4 \textwidth}
      \begin{lpic}{qq_nolabel_625720(0.30)}
      \end{lpic}
      \subcaption*{\hspace{0.1cm}4a. AB Inbev NV, market clearing}
    \end{subfigure}
    ~
    \begin{subfigure}[t]{0.4 \textwidth}
      \begin{lpic}{qq3_625720(0.30)}
      \end{lpic}
      \subcaption*{\hspace{0.1cm}4b. AB Inbev NV, log-normal}
    \end{subfigure}
    \\
    \begin{subfigure}[t]{0.4\textwidth}
      \begin{lpic}{qq_nolabel_68(0.30)}  
      \end{lpic}
      \subcaption*{\hspace{0.1cm}5a. Deutsche Telekom AG, market clearing}
    \end{subfigure}
    ~
    \begin{subfigure}[t]{0.4\textwidth}
      \begin{lpic}{qq3_68(0.30)}  
      \end{lpic}
      \subcaption*{\hspace{0.1cm}5b. Deutsche Telekom AG, log-normal}
    \end{subfigure}
    \caption{\label{fig:qq}QQ-plots of the samples of quantiles of closing prices (vertical axis) against theoretical $U[0,1]$-quantiles(horizontal axis), for the market clearing model and the log-normal model, for 5 Eurostoxx 50 index constituents.}
  }
\end{figure}
One could wonder how the market clearing model performs if the extreme beta distribution for order flow imbalance is replaced by Poisson order flow (this essentially corresponds to the call auction model of \cite{Toke15}). To investigate this option, we performed exactly the same analysis using the estimated clearing price distribution (as in definition \ref{def:estcpdist}) to obtain a sample of quantiles, but now with Poisson order flow: $(N_A,N_B) \sim \text{Pois}(50) \times \text{Pois}(50)$. Figure \ref{fig:tokeqq} shows the corresponding QQ-plot for one of the stocks (Airbus SE), similar results are obtained for the other stocks. The Poisson order flow, expressing the conviction that the orderflow imbalance $\alpha$ does not display great stochastic fluctuation around $\alpha=\frac12$, leads to a clearing price distribution that extremely underestimates the tails (even worse than the log-normal model). It turns out that the extreme order flow distributions are necessary to capture the tails of closing price distributions, underlining the limitations of Poisson order flow.
\begin{figure}[hbt]
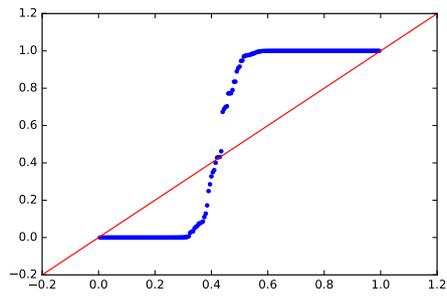

  \centering
  \parbox{13.5cm}{
  \centering
      \begin{lpic}{qqToke_85691(0.44)}     
      \end{lpic}
      
   \caption{\label{fig:tokeqq} QQ-plot of the sample of quantiles of closing prices of Airbus SE (vertical axis) against theoretical $U[0,1]$-quantiles(horizontal axis), for the market clearing model with $(N_A,N_B)\sim\text{Pois}(50) \times \text{Pois}(50)$.}
   }
\end{figure}

We conclude that the detailed shape of estimated clearing price
distributions from the market clearing model (with peaks and
troughs as in the right panel of figure~\ref{fig:realprices})
is informative for the relation between the price of an order and the corresponding execution
probability, while uni-modal shapes like those of the
log-normal distribution are not. Furthermore, we conclude that Poisson order flow does not display enough stochastic fluctuation to capture the tails of the observed randomness in closing prices, emphasizing the relevance of extreme order flow distributions. 
%%%%%%%%%%%%%%%%%%%%%%%%%%%%%%%%%%%%%%%%%%%%%%%%%%%%%%%%%%%%%%%%%%%%%%%%%%%%%%%%
\section{Conclusions}
\label{sec:conclusion}

In this article we propose a model for
auction price distributions in standard call
auctions based on a balance between two samples of random
orders. The model assumes \iid\ samples of buy- and sell-orders,
placed following demand- and supply-side valuation
distributions. An equilibrium equation (fixing the clearing
price by requiring that the number of buyers equals
the number of sellers) then leads to a
distribution for clearing price and transacted volume.
Bid- and ask-side volumes are left as free parameters (order flow); a choice for the distribution of
these parameters (possibly heavy-tailed or very skewed) leads to distributions for clearing
prices and transacted volumes, with or without a
limit order book.

In the highly liquid auctions of section~\ref{sec:asymp},
the clearing price distribution converges to a normal
central limit, with mean and variance in terms of
supply/demand-valuation distributions and order flow imbalance.
Most importantly, the variance of the limiting normal
distribution at real equilibrium price $x$ is inversely
proportional to the density of orders around $x$. The
interpretation is in regions on the price axis where price variance is
suppressed due to density of orders.

In subsection~\ref{sec:liqdist}, we consider the influence
of order flow on clearing price distributions.
Restriction to models involving Poisson or binomial assumptions
concerning the amount of liquidity on offer is hard to justify.
As confirmed empirically in section~\ref{sec:application}, extreme or skewed
order flow conditions are equally important.
Section~\ref{sec:supdem} explores the influence of
valuation distributions with some illustrative simulations:
for example, bringing valuation distributions closer together
increases transacted volume and decreases price variance. Closer
inspection of the price/volume distribution reveals that
there are two fundamentally different types of price variance, one driven by illiquidity and the other
by valuation uncertainty among market participants. To differentiate, one must involve
transacted volume, which is moderate when auction price
variance is minimal, low in illiquid markets and high in
markets with valuation uncertainty-driven price variance.

In section~\ref{sec:impact}, we analyse the model's
description of market impact. Remarkably, the model produces
a concave price impact function, especially when the valuation
distributions are widely separated, reflecting a market
in which the consensus is low. This is in line with
empirical results \cite{Hasbrouck91,Lilloetal03,Donierbonart15} and with the
theoretical results of \cite{Smithetal03}.

To statistically verify the validity of the model and
estimates of the daily closing price distributions
in section~\ref{sec:application},
we predict a year's worth of daily closing-price distributions
for five constituents of the Eurostoxx 50 index; Kolmogorov-Smirnov
statistics and QQ-plots demonstrate with ample statistical
significance that the model predicts
closing price distributions accurately, and compares
favourably with a simpler, log-normal, alternative
method of prediction. We conclude that the model's predicted clearing price distributions explain the observed randomness in closing prices well, confirming that the proposed model provides a proper description of price formation in call auctions.
%%%%%%%%%%%%%%%%%%%%%%%%%%%%%%%%%%%%%%%%%%%%%%%%%%%%%%%%%%%%%%%%%%%%%%%%%%%%%%%

\appendix

\section{Notation and proofs}
\label{app:proofs}

\subsubsection*{Notation and conventions}
We denote the multinomial coefficient for $n \geq 3$ by
\[
  {n \choose k_1,\dots,k_n} := \frac{n!}{k_1! \dots k_n!}.
\]
The binomial distribution with parameters $n$ and $p$ is denoted
$\Bin(n,p)$, the Poisson distribution with parameter $\lambda$
is denoted by $\text{Pois}(\lambda)$, the uniform distribution
on $[0,1]$ is denoted by $U[0,1]$ and the normal distribution
with mean $\mu$ and variance $\sigma^2$ is denoted by
$N(\mu,\sigma^2)$ with cumulative distribution function
$\Phi_{\mu,\sigma}(\cdot)$. Convergence in distribution is denoted $\conv{w.}$.  Let $\mathcal{X}\subset \RR$ be
the \emph{price-axis} which can be either discrete or
continuous. The lowest possible price is denoted by
$x_0:=\inf \calX$. The valuation distributions for supply and
demand prices, denoted $F_A$ and $F_B$, are assumed to be
distributions on the price-axis. 

\subsubsection*{Proofs}
The expressions we derive for price and price-volume distributions hinge
on the following two lemmas, which convert finding a solution to eq.
(\ref{eq:eq_equation}) into a question involving binomial distributions.
\begin{lemma}
\label{lem:technlemma}
For any $x \in \RR$, we have the equivalence:
$X \leq x \Leftrightarrow \DA(x) \geq \DB(x)$.
\end{lemma}
\begin{proof}
The left implication follows immediately from the definition of $X$, so
suppose $X \leq x$. Note that $x\mapsto\DA(x)$ is non-decreasing
and $x\mapsto\DB(X)$ is non-increasing. So the set
$\{y \in \RR : \DA(y) \geq \DB(y)\}$ is of the form $(a,\infty)$
or $[a,\infty)$, for some $a \in \RR$.
Through their definitions, $\DA$ and $\DB$ are right-continuous, so we
can write,
\[
  \DA(a) = \lim_{z \downarrow a} \DA(z) \geq \lim_{z \downarrow a} \DB(z)
  = \DB(a).
\]
Therefore $\{y \in \RR : \DA(y) \geq \DB(y)\}= [a,\infty)$, which implies that
$a = \inf\{y \in \RR : \DA(y) \geq \DB(y)\} \leq x$. Hence $x \in [a,\infty)
= \{y \in \RR : \DA(y) \geq \DB(y)\}$, which proves the result. 
\end{proof}
The independence assumption for $(A_1,\dots,A_{N_A})$ and $(B_1,\dots,B_{N_B})$ directly implies the content of the following lemma.
\begin{lemma}\label{lem:bindist}
For every $x \in \calX$, $(\DA(x),\DB(x))$ is a pair of independent,
binomially distributed random variables,
\begin{equation}
  \label{eq: bindist}
  (\DA(x), \DB(x)) \sim \Bin(N_A, F_A(x))\times \Bin(N_B, 1-F_B(x)).
\end{equation}
\end{lemma}
These two lemmas imply the following explicit expression for the clearing
price distribution in terms of the distributions of supply and demand,
$F_A$ and $F_B$, conditional on $N_A$ and $N_B$.
\thmcpdist*
\begin{proof}
From lemma~\ref{lem:technlemma} and the independence of $\DA(x)$ and $\DB(x)$
it follows that,
\[
  \begin{split}
    \PP(X \leq x ) &= \PP(\DA(x) \geq \DB(x))
      = \sum_{k=0}^{N_A} \PP(\DB(x) \leq k |\DA(x) = k)\PP(\DA(x) = k) \\ 
    & = \sum_{k=0}^{N_A} \sum_{l=0}^{N_B \wedge k} \PP(\DB(x) = l) \PP(\DA(x) = k),
  \end{split}
\]
where conditioning on $N_{A}, N_{B}$ has been omitted for ease of notation. 
The result follows from lemma~\ref{lem:bindist}.
\end{proof}
Similarly, we derive the joint distribution of $(X,V)$ from
eq.(\ref{eq:eq_equation}). Recall that the price-axis $\calX$ is
a discrete set, $\calX :=\{x_0, x_0 + \delta, \dots\}$, for some
$\delta>0$.
\thmxv*
\begin{proof}
In order to characterize the transacted volume $V$ in a similar sense as the clearing price in Lemma \ref{lem:technlemma}, define the generalized inverses $\DA^{-1}$ and $\DB^{-1}$ of $\DA$ and $\DB$ by 
\begin{align*}&\DA^{-1}(v) := \FA^{-1}(v/N_A) = \inf\{x \geq 0 : \DA(x) \geq v\},& \\
&\DB^{-1}(v) := \FB^{-1}(1-v/N_B) = \inf\{x \geq 0 : \DB(x) \leq v\}, \hfill&
\end{align*}
where $\FA^{-1},\FB^{-1}$ are the generalized inverses of the empirical cumulative distribution functions $\FA,\FB$ (for a distribution function $F$, its generalized inverse is defined as $F^{-1}(p) = \inf \{x \in \RR: F(x)\geq p\}$, for $p \in [0,1]$, see \eg\ \cite{Vaart98}, Chapter 21). For a given distribution function $F$, its generalized inverse satisfies
$$F^{-1}(p) \leq x \Leftrightarrow p \leq F(x),$$ 
which implies
\begin{equation}
\label{eq: DADB_inv_ineq} \DA^{-1}(v) \leq x \Leftrightarrow v \leq \DA(x),~\DB^{-1}(v) \leq x \Leftrightarrow v \geq \DB(x). 
\end{equation}
It follows from equation (\ref{eq: DADB_inv_ineq}) that $V$ is characterized by the following equivalences.
\begin{equation} \label{eq: v_equiv}
V \leq v \Leftrightarrow \DA(X) < v+1 \Leftrightarrow X < \DA^{-1}(v+1)\Leftrightarrow X \leq \DA^{-1}(v+1) -\delta,
\end{equation}
which leads to
$$ X \leq x, V \leq v \Leftrightarrow X \leq \min(x,\DA^{-1}(v+1) -\delta).$$
So we can write
\begin{equation} \label{eq: twodiffterms_xv} 
 \begin{split} 
\PP(X \leq x, V \leq v)& = \PP(X \leq \min(x,\DA^{-1}(v+1) -\delta))  \\
&=\PP(X \leq x, \DA^{-1}(v+1)-\delta >x) \\
&\qquad + \PP(X \leq \DA^{-1}(v+1)-\delta, \DA^{-1}(v+1)-\delta \leq x)  \\
& =  \PP(X \leq x, \DA^{-1}(v+1) >x+\delta) + \sum_{y \in \mathcal{X}, y \leq x}\PP(X \leq y, \DA^{-1}(v+1) =y+\delta). 
\end{split}
\end{equation}
Here, and in remainder of the proof, we have omitted the conditioning on $N_A,N_B$ in the notation, for convenience.
We start with the first term in this expression. From lemma \ref{lem:technlemma} and (\ref{eq: DADB_inv_ineq}) it follows that
\begin{align*} \PP(X \leq x, \DA^{-1}(v+1) >x+\delta) & = \PP(\DA(x) \geq \DB(x), \DA(x+\delta) < v+1) \\
& = \sum_{u=0}^v  \PP(\DA(x) \geq \DB(x), \DA(x+\delta) = u),
\end{align*}
where, by independence of the bid- and ask-samples,
\begin{equation} \label{eq: equation_part_xv}
\begin{split}
 \PP(\DA(x) &\geq \DB(x), \DA(x+\delta) = u) \\
 &= \sum_{k=0}^u \PP(\DB(x) \leq k) \PP(\DA(x) = k, \DA(x+\delta)=u)  \\
& = \sum_{k=0} ^u \sum_{l=0}^k \left[ { N_B \choose l}(1-F_B(x))^l F_B(x)^{N_B-l}  \times \frac{N_A!}{k!(u-k)!(N_A-u)!}  F_A(x)^k \right.  \\
&\qquad \qquad \times(F_A(x+\delta)-F_A(x))^{u-k}(1-F_A(x+\delta))^{N_A-u} \left.\vphantom{\frac12} \right],
\end{split}
\end{equation}
which gives the first term of the solution in equation (\ref{eq: result_xvdist}). \\
Now we turn to the second term in equation (\ref{eq: twodiffterms_xv}), for which we write
\begin{equation} \label{eq: equation_part_xv2} \begin{split} 
\sum_{y \in \mathcal{X}, y \leq x} &\PP(X \leq y, \DA^{-1}(v+1) =y+\delta) \\
&= \sum_{y \in \mathcal{X}, y \leq x} \PP(X \leq y, \DA^{-1}(v+1) \leq y+\delta)- \sum_{y \in \mathcal{X}, y \leq x} \PP(X \leq y, \DA^{-1}(v+1) \leq y) \\
&= \sum_{y \in \mathcal{X}, y \leq x} \PP(\DA(y) \geq \DB(y), \DA(y+\delta) \geq v+1) \\
&\qquad-\sum_{y \in \mathcal{X}, y \leq x}  \PP(\DA(y) \geq \DB(y), \DA(y) \geq v+1),
\end{split}
\end{equation}
where the last line follows by  Lemma \ref{lem:technlemma} and equation (\ref{eq: DADB_inv_ineq}).
The first term of this expression equals 
\begin{equation*}\sum_{y \in \mathcal{X}, y \leq x} \PP(\DA(y) \geq \DB(y), \DA(y+\delta) \geq v+1) = \sum_{y \in \mathcal{X}, y \leq x} \sum_{u=v+1}^{N_A} \PP(\DA(y) \geq \DB(y), \DA(y+\delta) = u)
\end{equation*}
 and its expression follows from equation (\ref{eq: equation_part_xv}), by substituting $y$ for $x$. This gives the second term of the solution in equation (\ref{eq: result_xvdist}). Finally, consider the second term in equation (\ref{eq: equation_part_xv2}), which equals
\begin{equation*}
\sum_{y \in \mathcal{X}, y \leq x}  \PP(\DA(y) \geq \DB(y), \DA(y) \geq v+1) = \sum_{y \in \mathcal{X}, y \leq x} \sum_{l=0}^{N_B}  \PP(\DA(y) \geq \max(l,v+1))\PP(\DB(y) = l),
\end{equation*}
by independence of the bid- and ask-samples. Using the binomial distributions of $\DA(y)$ and $\DB(y)$ once more, we see that this equals the last term in the solution of equation (\ref{eq: result_xvdist}), which concludes the proof.
\end{proof}

\highliqlimit*
\begin{proof}
The assumption that $F_A$ and $F_B$ are continuous implies that the steps
of $\DA$ and $\DB$ all have size 1, almost surely. So we have, almost surely,
\[
  N_A \FA(X) = N_B(1-\FB(X))+\Delta(X).
\]
Combination with (\ref{eq:real_eq_price}) yields,
\[
  N_A(\FA(X) - F_A(x_E)) = -N_B(\FB(X) -F_B(x_E)) +\Delta(X), 
\]
which, after introduction of $F_A(X)$ and $F_B(X)$, reads,
\begin{equation}
\label{eq:proof_highliqlimit1}
\begin{split}
  \sqrt{\frac{N_A}{N_B}}(&\FA(X) - F_A(X))
    + \sqrt{\frac{N_A}{N_B}}(F_A(X) - F_A(x_E))\\
  & = -\sqrt{\frac{N_B}{N_A}}(\FB(X) - F_B(X))
    -\sqrt{\frac{N_B}{N_A}}(F_B(X) - F_B(x_E))+ \frac{\Delta(X)}{\sqrt{N_AN_B}}.
\end{split}
\end{equation}
Now denote 
\[
  Z_{A,N_A}(x) = \sqrt{N_A}(\FA(x) - F_A(x)), \quad
  Z_{B,N_B}(x) = \sqrt{N_B}(\FB(x) - F_B(x)).
\]
By Donsker's theorem (see \eg\ \cite{Vaart98}, Theorem 19.3) and independence
of the bid- and ask-samples, it holds that
\[
  (Z_{A,N_A}(x),Z_{B,N_B}(x)) \conv{w.}
  N\bigl(0,F_A(x)(1-F_A(x))\bigr)\times N\bigl(0,F_B(x)(1-F_B(x))\bigr),
\]
as $N_A,N_B \to \infty$, uniformly over $x \in \RR$ (and hence for every
random $X$). Using $N_A = \alpha N, N_B=(1-\alpha)N$ and
$D(x) = \Delta(x)/\sqrt{N}$, we can rewrite (\ref{eq:proof_highliqlimit1})
as follows,
\[
  \begin{split}
  \sqrt{\frac{\alpha}{1-\alpha}}(F_A(X)&-F_A(x_E))
    + \sqrt{\frac{1-\alpha}{\alpha}}(F_B(X)-F_B(x_E))\\
  &= -\frac{1}{\sqrt{N(1-\alpha)}} Z_{A,N_A}(X)
    - \frac{1}{\sqrt{\alpha N}} Z_{B,N_B}(X)
    + \frac{D(X)}{ \sqrt{N(1-\alpha)\alpha}}.
  \end{split}
\]
Hence, we obtain the following weak limit,
\begin{equation}
\label{eq:proof_highliqlimit2}
  \frac{\sqrt{N}}{\tau(X)}\biggl(\alpha(F_A(X)
  -F_A(x_E))
    + (1-\alpha)(F_B(X) - F_B(x_E))\biggr)
    -\frac{D(X)}{\tau(X)} \conv{w.} N(0,1),
\end{equation}
where the asymptotic variance $\tau^2(X)$ is given by,
\[
  \tau^2(X) = \alpha\,F_A(X)(1-F_A(X))+(1-\alpha)F_B(X)(1-F_B(X)).
\]
With the help of the distribution function $F_R$, defined by
the convex combination,
\[
  F_R(\cdot) = \alpha\,F_A(\cdot) + (1-\alpha)F_B(\cdot),
\]
we rewrite equation (\ref{eq:proof_highliqlimit2}) as follows,
\[
  \frac1{\tau(X)}\bigl(\sqrt{N}(F_R(X)-F_R(x_E))
  - D(X)\bigr) \conv{w.} N(0,1).
\]
Since $0<\tau(X)<1$ and $D$ is bounded, we conclude that
$F_R(X)$ converges to $F_R(x_E)$ in probability. The assumptions on
$F_A$ and $F_B$ ensure that $F_R$ has a Lebesgue density $f_R$ and
that $F_R$ is invertible with continuous inverse $F_R^{-1}:[0,1]\to\RR$,
so it follows that $X$ converges to $x_E$ in probability. By continuity
it follows that $\tau(X)$ converges in probability to $\tau(x_E)$ and
$D(X)$ to $D(x_E)$. By Slutsky's Lemma (see \eg\ \cite{Vaart98}, Lemma 2.8), we arrive at,
\[
  \sqrt{N}(F_R(X)-F_R(x_E)) \conv{w.} N(D(x_E),\tau^2(x_E)).
\]
The Delta-method (see \eg\ \cite{Vaart98}, Theorem 3.1) then leads to,
\[
  \sqrt{N}(X-x_E) \conv{w.} (F_R^{-1})'(F_R(x_E)) N(D(x_E),\tau^2(x_E)),
\]
where, according to the inverse function theorem,
\[
  (F_R^{-1})'(F_R(x_E)) = \frac{1}{f_R(x_E)}.
\]
\end{proof}
%%%%%%%%%%%%%%%%%%%%%%%%%%%%%%%%%%%%%%%%%%%%%%%%%%%%%%%%%%%%%%%%%%%%%%%%%%%%%%%\\
\subsection*{{\normalfont \textbf{Additional Figures}}}
\begin{figure}[h]
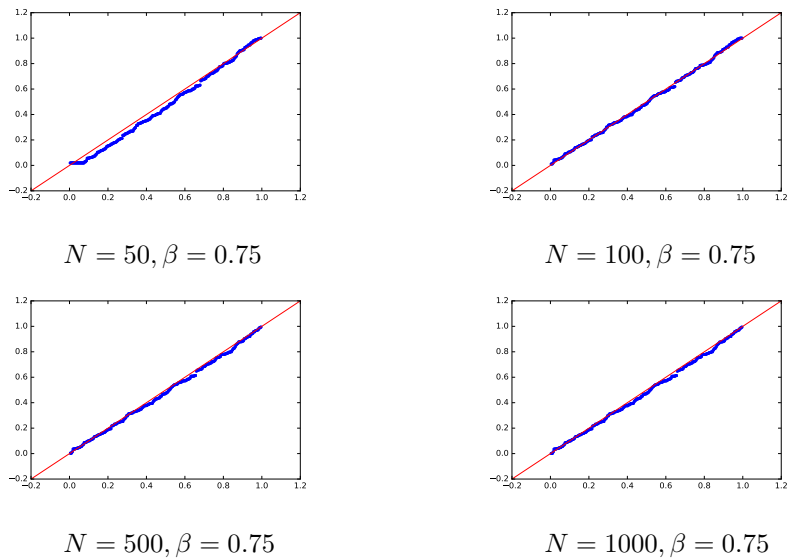

\captionsetup[subfigure]{justification=justified,singlelinecheck=false}
  \centering
  \parbox{0.85\textwidth}{
    \begin{subfigure}[t]{0.4\textwidth}
      \begin{lpic}{qq2_68-50-0.75(0.30)}  
      \end{lpic}
      \subcaption*{\hspace{1cm}$N=50, \beta=0.75$}
    \end{subfigure}
        \begin{subfigure}[t]{0.4\textwidth}
      \begin{lpic}{qq_68-100-0.75(0.30)}  
      \end{lpic}
      \subcaption*{\hspace{1cm}$N=100, \beta=0.75$}
    \end{subfigure}

        \begin{subfigure}[t]{0.4\textwidth}
      \begin{lpic}{qq2_68-500-0.75(0.30)}  
      \end{lpic}
      \subcaption*{\hspace{1cm}$N=500, \beta=0.75$}
    \end{subfigure}
        \begin{subfigure}[t]{0.4\textwidth}
      \begin{lpic}{qq2_68-1000-0.75(0.30)}  
      \end{lpic}
      \subcaption*{\hspace{1cm}$N=1000, \beta=0.75$}
    \end{subfigure}
    
    \caption{\label{fig:roleN} How the choice of $N$ affects the results of section \ref{sec:application}, for the case of Deutsche Telekom AG (similar effects are observed for the other stocks). It is seen that the choice of $N$ does not really impact the results, as long as $N$ is sufficiently large ($N>50$). }}
\end{figure}
\begin{figure}[b]
  \centering
\parbox{0.85\textwidth}{
\captionsetup[subfigure]{justification=justified,singlelinecheck=false}
  \parbox{0.85\textwidth}{
        \begin{subfigure}[t]{0.4\textwidth}
      \begin{lpic}{qq2_68-100-0.65(0.30)}  
      \end{lpic}
      \subcaption*{\hspace{1cm}$N=100, \beta=0.65$}
    \end{subfigure}
        \begin{subfigure}[t]{0.4\textwidth}
      \begin{lpic}{qq2_68-100-0.7(0.30)}  
      \end{lpic}
      \subcaption*{\hspace{1cm}$N=100, \beta=0.7$}
    \end{subfigure}
    
        \begin{subfigure}[t]{0.4\textwidth}
      \begin{lpic}{qq_68-100-0.75(0.30)}  
      \end{lpic}
      \subcaption*{\hspace{1cm}$N=100, \beta=0.75$}
    \end{subfigure}
        \begin{subfigure}[t]{0.4\textwidth}
      \begin{lpic}{qq2_68-100-0.8(0.30)}  
      \end{lpic}
      \subcaption*{\hspace{1cm}$N=100, \beta=0.8$}
    \end{subfigure}
    
    \centering \hspace{-7.5mm}
        \begin{subfigure}[t]{0.4\textwidth}
      \begin{lpic}{qq2_68-100-0.85(0.30)}  
      \end{lpic}
      \subcaption*{\hspace{1cm}$N=100, \beta=0.85$}
    \end{subfigure} 
    }
    \caption{\label{fig:rolebeta} How the choice of $\beta$ affects the results reported in section \ref{sec:application}, for the case of Deutsche Telekom AG (similar effects are observed for the other stocks). It is seen that the results are robust with respect to the choice of $\beta \in (0.65,0.85)$.}}
\end{figure}

\end{document}